\documentclass[reprint,aps,pre,longbibliography]{revtex4-2}

\usepackage{amsmath,amssymb,mathtools,bm}
\usepackage{amsthm}
\usepackage{graphicx}
\usepackage{xcolor}
\usepackage{siunitx}
\usepackage{hyperref}
\usepackage{physics}
\usepackage{enumitem}

\newcommand{\gradx}{\nabla_{\!x}}
\newcommand{\grads}{\nabla_{\!s}}
\newcommand{\grady}{\nabla_{\!y}}
\newcommand{\divv}{\nabla\!\cdot}
\newcommand{\gradG}{\nabla_{\!\Gamma}}
\newcommand{\divG}{\nabla_{\!\Gamma}\!\cdot}
\newcommand{\K}{\mathbf K}

\newcommand{\vt}{\mathbf v_t}
\newcommand{\valpha}{\mathbf v_\alpha}

\newcommand{\pg}{p_g}
\newcommand{\Lam}{\Lambda_t}
\newcommand{\lam}{\lambda}
\newcommand{\s}{\bm s}

\newcommand{\poro}{\phi}
\newcommand{\RR}{\mathbb R}
\newcommand{\cT}{\Theta}
\newcommand{\cH}{\mathcal H}
\newcommand{\cQ}{\mathcal Q}
\newcommand{\cD}{\mathcal D}

\newcommand{\Tmf}{\mathcal T^{mf}}

\newtheorem{theorem}{Theorem}

\newtheorem{corollary}{Corollary}
\newtheorem{definition}{Definition}

\begin{document}

\title{Nonisothermal global-pressure exactness in fractured multiphase flow with aperture feedback}

\author{Christian Tantardini}
\email{christiantantardini@ymail.com}
\affiliation{Center for Integrative Petroleum Research, King Fahd University of Petroleum and Minerals, Dhahran 31261, Saudi Arabia}

\author{Fernando Alonso-Marroqu\'in}
\email{fernando@quantumfi.net}
\affiliation{Department of Computational Physics for Engineering Material, ETH Zurich, 8092 Zurich, Switzerland}
\affiliation{Center for Integrative Petroleum Research, King Fahd University of Petroleum and Minerals, Dhahran 31261, Saudi Arabia}
\date{\today}

\begin{abstract}
Global-pressure formulations recast multiphase Darcy flow in terms of a single pressure driving the total flux. Their exact equivalence to phase-pressure formulations holds only when the constitutive data satisfy the compatibility conditions required for a total-differential structure and its generalized nonisothermal extension. Here, we derive the exactness criterion for temperature-dependent mobilities and capillary pressures. We show that equivalence depends on whether the mobility-weighted capillary contribution is path independent in the saturation--temperature domain, so that it can be absorbed into a scalar global pressure. This yields the classical compatibility conditions within the saturation sector and a distinct mixed saturation--temperature condition that arises only in nonisothermal settings. We then incorporate this structure into a reduced matrix--fracture model with heat transport, matrix--fracture thermal exchange, and evolving aperture. Numerical benchmarks recover the three regimes predicted by the theory: globally exact, exact on each fixed-temperature slice but not on the full saturation--temperature domain, and fully nonexact. In fractured systems, thermal forcing alone can drive transitions between these regimes, while aperture evolution changes the path through state space. When saturation-sector exactness is lost, a least-squares projection on fixed-temperature slices extracts the nearest gradient component of the mobility-weighted capillary field. This yields a conservative slice-wise scalar-pressure surrogate and a quantitative projection residual. The residual separates saturation-sector nonintegrability from the mixed saturation--temperature incompatibility that controls genuinely nonisothermal loss of exactness. The framework links nonisothermal exactness theory, fractured-flow dynamics, and conservative reduced closure in a global-pressure formulation.
\end{abstract}

\maketitle

\section{Introduction}
\label{sec:intro}

Immiscible displacement in porous and fractured media is governed by the coupling between pressure-driven transport, capillary redistribution, and geometry-controlled flow localization. The Buckley--Leverett formulation isolates this structure by
combining a pressure-driven total flux with fractional-flow transport
\cite{BuckleyLeverett1942,Leverett1941}. Recent studies have extended this picture in complementary directions: capillary-number-dependent upscaling explains how heterogeneity moves Buckley--Leverett dynamics between viscous and capillary
regimes \cite{BenhamBickleNeufeld2021JFM}; upscaled pressure-difference theory shows that macroscopic capillary pressure contains hydrodynamic and interfacial
contributions beyond a saturation-only law \cite{LasseuxValdesParada2023JFM};
and depth-integrated fracture modelling shows how aperture variation, wall drag, and out-of-plane curvature control immiscible displacement in rough fractures
\cite{KrishnaMeheustNeuweiler2025JFM}. A central challenge is to formulate the corresponding nonisothermal theory for regimes in which multiphase displacement,
fracture-aperture evolution, matrix--fracture thermal exchange, and
temperature-dependent capillary laws act simultaneously.

For multiphase Darcy flow, one of the most useful reductions is the
global-pressure formulation. Instead of solving for all phase pressures
independently, one introduces a scalar pressure whose gradient drives the total
volumetric flux, while capillary effects enter as mobility-weighted drift
terms. For two-phase flow this viewpoint is classical, and fully equivalent
compressible variants are also available
\cite{ChaventJaffre1986,AmazianeJurak2008,Amaziane2011}. Nonisothermal
two-phase fractional-flow formulations have also been developed using
saturation, nonisothermal global pressure, and temperature as primary variables
\cite{Amaziane2024CVFENonisothermal}. For three-phase flow,
however, exact equivalence with the phase-pressure formulation is no longer
automatic. The extension of the two-phase picture is mathematically delicate:
three-phase Buckley–Leverett systems can lose strict hyperbolicity and develop
umbilic or mixed-type behavior \cite{Holden1990}. At the constitutive level,
exact reduction requires compatibility between mobilities and capillary data.
In the isothermal setting, this requirement is expressed by the
Total-Differential (TD) condition and its generalized Total-Differential
(gTD) extensions
\cite{Chavent2009,DiChiara2010}. When these conditions hold, the
mobility-weighted capillary contribution can be absorbed into a scalar
potential on the saturation simplex, so the associated integral becomes path
independent. Constructing physically plausible three-phase data with this
property remains a nontrivial constitutive problem
\cite{DiChiara2010,Schaefer2020}. Comparison studies have also shown why the
global-pressure formulation is computationally attractive when such
compatibility is available \cite{ChenEwing1997}.

Our recent Global Buckley–Leverett framework for multicomponent flow in
fractured media was developed within this isothermal setting. It combines
conservative global-pressure reduction with equation-of-state coupling,
fractured-media transport, and dynamic capillarity \cite{nfhq-z872}, while retaining the central
role of Total-Differential compatibility in determining whether exact pressure
decoupling is available
\cite{tantardini2025globalbuckleyleverettmulticomponentflow}. The present work
extends this framework to nonisothermal fractured multiphase flow by treating
saturation and temperature as a coupled thermodynamic state.

In nonisothermal multiphase flow, viscosities, densities, mobilities, and
capillary laws depend on temperature. Thermodynamic analyses and
experiments also show that capillary pressure itself can vary with temperature
\cite{HassanizadehGray1993,SheSleep1998,Jurak2019,Standnes2021,Miller2019}.
The relevant state domain is therefore no longer the saturation simplex alone,
but an augmented domain in which saturation and temperature evolve together.

This shift changes the exactness question in an essential way. In the
isothermal theory, one asks whether the mobility-weighted capillary
contribution is integrable on saturation domain. In the nonisothermal theory,
one must ask whether the corresponding contribution is integrable on the full
saturation--temperature state domain. The classical compatibility requirement
on each fixed-temperature saturation slice remains part of the answer, but it
is no longer sufficient. A system may remain exact on every fixed-temperature
slice and still fail to admit a globally exact nonisothermal reduction once
thermal evolution is included. Conversely, the loss of exactness can be
diagnosed by separating saturation-sector incompatibility from genuinely mixed
saturation–temperature incompatibility.

The key observation is that nonisothermal global-pressure exactness is governed
by the coupled saturation--temperature state domain as a whole; it cannot be
deduced solely from exactness tests on individual fixed-temperature saturation
slices. Thus, the novelty is not the general idea of global pressure or the
standard calculus fact that a conservative vector field admits a scalar
potential. Rather, the new contribution is the identification of the
mobility-weighted capillary field appropriate to temperature-dependent
multiphase capillary data and the resulting mixed compatibility condition. The
classical isothermal condition reappears as the restriction of the broader
criterion to fixed-temperature slices, while the additional mixed compatibility
determines whether those slice-wise potentials can be assembled into a single
nonisothermal global-pressure potential. The precise augmented-state criterion
is developed in Sec.~\ref{sec:thermal_gtd}.

We place this structure in a fractured setting motivated by geothermal and
thermo-hydraulic reservoir applications. Mixed-dimensional and reduced fracture
models are widely used because they retain matrix--fracture exchange while
avoiding the cost of resolving thin fractures explicitly
\cite{Berre2019,PruessNarasimhan1985,Ahmed2017}. For geothermal applications,
these models must also capture coupled mass and heat transport together with
matrix--fracture thermal exchange
\cite{Wang2021FracturedGeothermal,Ding2022}. At larger field scale, discrete
fracture network descriptions are often required because geothermal reservoirs
are anisotropic and heterogeneous, and fracture orientation, density,
connectivity, transmissivity, and aperture control the effective flow pathways
\cite{Medici2023DFNGeothermal,Lei2023EGSDFNTHM}. At the fracture scale,
transmissivity is highly sensitive to aperture and, in the smooth-fracture
limit, scales cubically with it \cite{Witherspoon1980}.
Under thermal injection
and production, thermo-poroelastic stresses can alter aperture and thereby
modify permeability and channeling \cite{Ghassemi2008,Pandey2017,Bisdom2016}.
Temperature therefore affects the problem in two distinct ways: it changes the
constitutive conditions for exactness, and it changes the path by which the
coupled matrix–fracture system moves through state space.

The fluid-mechanics relevance of this coupling is strongest when thermal
fronts, capillary effects, matrix–fracture heat exchange, buoyancy, and
aperture-sensitive transmissivity act on comparable scales. In such regimes,
temperature does not merely modify material parameters passively; it reshapes
the local trajectory followed by the multiphase state. This trajectory
determines whether the global-pressure reduction remains exact, becomes only
slice-wise exact, or must be treated as a projected reduced closure.

The study has three objectives. First, we derive the nonisothermal extension of
the Total-Differential exactness criterion on the augmented state domain of saturation and temperature. Second, we embed this structure in a reduced
matrix–fracture thermal model with evolving aperture. Third, we introduce
diagnostics and a conservative slice-wise projection to quantify the loss of
exactness when the mobility-weighted capillary field is not globally
integrable. In nonexact regimes, the projected formulation is used as a
reduced conservative surrogate; it does not restore full equivalence to the
original phase-pressure formulation when the mixed nonisothermal defect is
present.

The remainder is organized as follows. Section~\ref{sec:gbl} reviews the
global-pressure structure and the isothermal Total-Differential (TD) condition and its generalized
Total-Differential (gTD) background \cite{Chavent2009,DiChiara2010}.
Section~\ref{sec:thermal_gtd} presents the nonisothermal exactness theorem on
the augmented state domain. Section~\ref{sec:matrix_fracture} embeds the theory
in a reduced matrix–fracture thermal system with evolving aperture.
Section~\ref{sec:projection} introduces the conservative slice-wise projection
surrogate for nonexact regimes.
Section~\ref{sec:numerics} presents numerical verification and diagnostic
validation benchmarks, including augmented-state diagnostics, fractured thermal
fronts, aperture feedback, and quantitative defect summaries.
Section~\ref{sec:conclusion} summarizes the main results and limitations.

\section{Global pressure and isothermal TD/gTD background}
\label{sec:gbl}

\subsection{Multiphase Darcy structure}

We begin from the Darcy-scale description of immiscible multiphase flow in
porous media, in a form that supports the later nonisothermal extension and
underlies both Buckley--Leverett-type reductions and global-pressure
formulations
\cite{BuckleyLeverett1942,Leverett1941,Bear1972,ChaventJaffre1986}. Let
$\Omega\subset \RR^d$ be a porous domain. For phases
$\alpha=1,\dots,n_p$, let $p_\alpha$ denote the phase pressures,
$S_\alpha$ the phase saturations, and $\lam_\alpha(S,T)\ge 0$ the phase
mobilities, typically written as $\lam_\alpha=k_{r\alpha}/\mu_\alpha$.
Here $S=(S_1,\dots,S_{n_p})$ denotes the full saturation state.

The saturations satisfy
\begin{equation}
S_\alpha\ge 0,
\qquad
\sum_{\alpha=1}^{n_p}S_\alpha=1.
\label{eq:sat_constraint}
\end{equation}
For three phases, it is convenient to introduce simplex coordinates
\begin{equation}
\s=(s_1,s_2)=(S_1,S_2),
\qquad
S_3=1-s_1-s_2,
\label{eq:simplex_coords}
\end{equation}
on the ternary saturation simplex
\begin{equation}
\Delta=\{\s\in\RR^2:\ s_1\ge 0,\ s_2\ge 0,\ s_1+s_2\le 1\}.
\label{eq:simplex_def}
\end{equation}

Ignoring fracture geometry for the moment, the phase Darcy velocities are
\begin{equation}
\valpha
=
-\K\,\lam_\alpha(S,T)
\Big(
\gradx p_\alpha-\rho_\alpha(T)\,g\,\gradx z
\Big),
\label{eq:phase_darcy}
\end{equation}
where $\K$ is the intrinsic permeability tensor, $\rho_\alpha$ is the phase
density, $z$ is elevation, and $g$ is the gravitational acceleration. The
total mobility and total Darcy velocity are
\begin{equation}
\Lam(S,T)=\sum_{\alpha=1}^{n_p}\lam_\alpha(S,T),
\qquad
\vt=\sum_{\alpha=1}^{n_p}\valpha.
\label{eq:total_defs}
\end{equation}
With a reference phase $r$ fixed, the capillary pressures are
\begin{equation}
p_{c,\alpha}(S,T)=p_\alpha-p_r,
\qquad
\alpha\neq r,
\qquad
p_{c,r}\equiv 0.
\label{eq:capillary_def}
\end{equation}

A global-pressure gradient $\gradx \pg$ is introduced through the
mobility-weighted identity
\begin{equation}
\sum_{\alpha=1}^{n_p}
\lam_\alpha(S,T)\,\gradx p_\alpha
=
\Lam(S,T)\,\gradx \pg,
\label{eq:pg_identity}
\end{equation}
which yields
\begin{align}
\vt
&=
-\K\left(\Lam(S,T)\,\gradx \pg
-
\sum_{\alpha=1}^{n_p}\lam_\alpha(S,T)\,\rho_\alpha(T)\,g\,\gradx z\right)
\nonumber \\
&=
-\K\,\Lam(S,T)
\left(\gradx \pg-\rho_t(S,T)\,g\,\gradx z\right),
\label{eq:vt_global}
\end{align}
with mobility-weighted density
\begin{equation}
\rho_t(S,T):=
\frac{1}{\Lam(S,T)}
\sum_{\alpha=1}^{n_p}
\lam_\alpha(S,T)\,\rho_\alpha(T).
\label{eq:rho_t}
\end{equation}

Equation~\eqref{eq:pg_identity} defines $\gradx \pg$ pointwise. Exactness is
a stronger statement: it asks whether this pointwise definition can be
represented by a scalar pressure correction depending only on the local
thermodynamic state.

Because the saturations satisfy \eqref{eq:sat_constraint}, the independent
state variables are the saturation coordinates on the simplex and the
temperature. In the nonisothermal setting we therefore work on a state domain
\begin{equation}
\mathcal U\subset \Delta\times\Theta,
\qquad
\Theta=[T_{\min},T_{\max}],
\label{eq:state_domain_intro}
\end{equation}
where $\Delta$ is the admissible saturation simplex and $\Theta$ is the
admissible temperature interval.

Throughout the paper, $\gradx$ denotes the spatial gradient in the matrix
domain, while $\gradG$ denotes the tangential spatial gradient along the
fracture. Gradients with respect to saturation variables are denoted by
$\grads$. For the augmented state variable
\begin{equation}
    y=(s_1,\ldots,s_{n_p-1},T),
\end{equation}
we write
\begin{equation}
\nabla_{(\s,T)}=\grady
=
(\partial_{s_1},\ldots,\partial_{s_{n_p-1}},\partial_T).
\end{equation}
Thus, if $\Pi(\s,T)$ is a scalar function on the state domain,  $\grady\Pi$ is a gradient in the saturation--temperature state domain,
not a spatial gradient. When the state variables depend on position,
$\s=\s(x,t)$ and $T=T(x,t)$, the spatial gradient of the composite function
$\Pi(\s(x,t),T(x,t))$ is obtained by the chain rule,
\begin{widetext}
\begin{equation}
\gradx[\Pi(\s(x,t),T(x,t))]
=
\sum_{i=1}^{n_p-1}
\partial_{s_i}\Pi(\s,T)\,\gradx s_i
+
\partial_T\Pi(\s,T)\,\gradx T .
\label{eq:chain_rule_Pi_intro}
\end{equation}
\end{widetext}
Using the reference phase $r$, write
\begin{equation}
p_\alpha=p_r+p_{c,\alpha}(\s,T),
\qquad
p_{c,r}\equiv 0 .
\label{eq:palpha_pc_intro}
\end{equation}
Then
\begin{equation}
\gradx p_\alpha
=
\gradx p_r
+
\sum_{i=1}^{n_p-1}
\partial_{s_i}p_{c,\alpha}(\s,T)\,\gradx s_i
+
\partial_T p_{c,\alpha}(\s,T)\,\gradx T .
\label{eq:phase_grad_state_intro}
\end{equation}
Substitution into \eqref{eq:pg_identity} gives
\begin{equation}
\gradx \pg
=
\gradx p_r
+
\sum_{i=1}^{n_p-1} A_i(\s,T)\,\gradx s_i
+
B(\s,T)\,\gradx T,
\label{eq:pg_gradient_state_intro}
\end{equation}
where
\begin{equation}
A_i(\s,T)=
\frac{1}{\Lam(\s,T)}
\sum_{\alpha=1}^{n_p}
\lam_\alpha(\s,T)
\partial_{s_i}p_{c,\alpha}(\s,T),
\label{eq:A_intro}
\end{equation}
with $i=1,\ldots,n_p-1,$
and
\begin{equation}
B(\s,T)=
\frac{1}{\Lam(\s,T)}
\sum_{\alpha=1}^{n_p}
\lam_\alpha(\s,T)
\partial_T p_{c,\alpha}(\s,T).
\label{eq:B_intro}
\end{equation}

In this work, \emph{exactness} means exact equivalence between the
global-pressure formulation and the corresponding phase-pressure formulation at
the continuous Darcy-model level. The phase-pressure formulation uses the phase
pressures \(p_\alpha\), which are related to the chosen reference phase
pressure \(p_r\) by Eq.~\eqref{eq:palpha_pc_intro}. Thus the differences
between phase pressures are carried by the capillary-pressure functions
\(p_{c,\alpha}\).

The global-pressure reduction is exact on \(\mathcal U\) if the
mobility-weighted capillary contribution in
Eq.~\eqref{eq:pg_gradient_state_intro} can be represented as the spatial
gradient of a scalar state function. More precisely, we require a scalar
potential
\begin{equation}
\Pi\in C^1(\mathcal U),
\label{eq:Pi_space_intro}
\end{equation}
where \(C^1(\mathcal U)\) means that \(\Pi\) is continuously differentiable
with respect to the independent saturation variables and temperature on the
state domain \(\mathcal U\). Hence the derivatives
\(\partial_{s_i}\Pi\) and \(\partial_T\Pi\) exist and are continuous. Exactness
requires
\begin{equation}
\partial_{s_i}\Pi(\s,T)=A_i(\s,T),
\qquad
\partial_T\Pi(\s,T)=B(\s,T),
\label{eq:exactness_intro}
\end{equation}
with \(i=1,\ldots,n_p-1\). Equivalently, in the augmented state coordinates
\(y=(s_1,\ldots,s_{n_p-1},T)\),
\begin{equation}
\grady \Pi(\s,T)
=
\big(A_1(\s,T),\ldots,A_{n_p-1}(\s,T),B(\s,T)\big).
\label{eq:gradPi_intro}
\end{equation}

Combining Eq.~\eqref{eq:exactness_intro} with the chain rule gives the spatial
capillary contribution
\begin{equation}
\gradx[\Pi(\s(x,t),T(x,t))]
=
\sum_{i=1}^{n_p-1}A_i(\s,T)\,\gradx s_i
+
B(\s,T)\,\gradx T .
\label{eq:gradPi_spatial_intro}
\end{equation}
When this condition holds, the global pressure is related directly to the
reference phase pressure by
\begin{equation}
\pg=p_r+\Pi(\s,T).
\label{eq:pg_ansatz_general}
\end{equation}
Here \(p_r\) is the pressure of the chosen reference phase, while
\(\Pi(\s,T)\) is the scalar capillary correction generated by the local
saturation--temperature state. Taking the spatial gradient of
Eq.~\eqref{eq:pg_ansatz_general} and using
Eq.~\eqref{eq:gradPi_spatial_intro} reproduces exactly the
mobility-weighted phase-pressure contribution in
Eq.~\eqref{eq:pg_gradient_state_intro}. This is the sense in which the
global-pressure formulation and the phase-pressure formulation are equivalent.

The capillary correction is then path independent in the
saturation--temperature state domain. Consequently, the pressure-driven part of
the total Darcy flux obtained from \(\pg\) is identical to that obtained from
the phase pressures. When Eq.~\eqref{eq:exactness_intro} fails, \(\pg\) may
still be used as a reduced pressure variable, but the formulation is no longer
exactly equivalent to the phase-pressure system.

This integrability issue lies at the core of global-pressure theory. For two
phases, the construction is classical \cite{ChaventJaffre1986}. For three
phases and beyond, exact decoupling is more delicate: the Buckley--Leverett
structure may lose strict hyperbolicity, and exact global-pressure equivalence
requires constitutive compatibility between mobilities and capillary laws
\cite{Holden1990,ChenEwing1997,Chavent2009,DiChiara2010}. Our recent Global
Buckley--Leverett formulation for multicomponent fractured flow was built on
this isothermal backbone: it extended the global-pressure/fractional-flow
viewpoint to equation-of-state coupling, dynamic capillarity\cite{nfhq-z872}, and
fractured-media transport, while retaining TD/gTD compatibility as the
criterion for exact pressure reduction
\cite{tantardini2025globalbuckleyleverettmulticomponentflow}. The present work
retains that Darcy-scale framework but reexamines the exactness criterion when
temperature becomes an active state variable.

\subsection{Isothermal TD/gTD}

The isothermal theory is recovered by fixing $T$. The state domain then
reduces from $\Delta\times\Theta$ to the saturation simplex $\Delta$, and
the temperature term in \eqref{eq:pg_gradient_state_intro} is absent. The
capillary field is therefore
\begin{equation}
\mathbf A(\s)
=
\big(A_1(\s),\ldots,A_{n_p-1}(\s)\big),
\label{eq:A_vector_iso}
\end{equation}
with
\begin{equation}
A_i(\s)=
\frac{1}{\Lam(\s)}
\sum_{\alpha=1}^{n_p}
\lam_\alpha(\s)
\partial_{s_i}p_{c,\alpha}(\s),
\qquad
i=1,\ldots,n_p-1 .
\label{eq:A_iso}
\end{equation}
An exact isothermal global-pressure reduction exists if there is a potential
$\Pi\in C^1(\Delta)$ such that
\begin{equation}
\grads\Pi(\s)=\mathbf A(\s),
\label{eq:gradPi_A_iso}
\end{equation}
or, equivalently,
\begin{equation}
\partial_{s_i}\Pi(\s)=A_i(\s),
\qquad
i=1,\ldots,n_p-1 .
\label{eq:Pi_iso}
\end{equation}
In that case,
\begin{equation}
\pg=p_r+\Pi(\s)
\label{eq:pg_iso_exact}
\end{equation}
is exactly equivalent to the phase-pressure formulation.

On simply connected subsets of the saturation simplex, the existence of
$\Pi$ is equivalent to the mixed-partial compatibility condition
\begin{equation}
\partial_{s_j}A_i=\partial_{s_i}A_j,
\qquad
1\le i,j\le n_p-1.
\label{eq:gtd_iso}
\end{equation}
The simply connected assumption excludes topological obstructions to defining a
single-valued global potential. The compatibility condition
\eqref{eq:gtd_iso} gives local path independence; on a simply connected
saturation domain it also guarantees that the local potentials can be assembled
into one global scalar correction $\Pi$. In the present applications this is a
mild restriction, since the regularized interior of the saturation simplex is
simply connected. This condition is stated explicitly because residual-saturation cutoffs,
phase-disappearance regions, or excluded nonregular subsets may change the
topology of the admissible state domain.

Condition~\eqref{eq:gtd_iso} is the standard calculus form of the
isothermal TD  and its gTD extension \cite{Chavent2009,DiChiara2010}. It states
that the mobility-weighted capillary field is a gradient field on the
saturation simplex, so that the capillary correction is path independent and
can be absorbed into a scalar global-pressure potential.

\section{Nonisothermal exactness on the augmented state domain}
\label{sec:thermal_gtd}

Section~\ref{sec:gbl} shows that, in the isothermal theory, exact
global-pressure reduction is an integrability statement on the saturation
simplex: the mobility-weighted capillary field must be the gradient of a scalar
saturation potential. The same viewpoint remains correct in the nonisothermal setting, but the relevant state domain is now described by the augmented variables $(\s,T)$ rather than by the saturation simplex alone. Let $\cT=[T_{\min},T_{\max}]\subset\mathbb R$ denote the admissible temperature range; (we reserve $\mathcal T_f(b)$ for the fracture transmissivity). Because mobilities, densities, viscosities, and capillary laws may all depend on temperature, and because capillary pressure is not, in general, a purely saturation-dependent quantity from a thermodynamic viewpoint, exactness is no longer determined by the saturation sector alone \cite{Chavent2009,DiChiara2010,HassanizadehGray1993,Standnes2021,Weishaupt2021,Jurak2019,SheSleep1998}. Even if the classical TD/gTD condition holds on every fixed-temperature slice, a global state potential may still fail to exist once thermal gradients are present. The new obstruction is the compatibility of the temperature derivative of the candidate potential with the mobility-weighted capillary contribution generated by $\partial_T p_{c,\alpha}$. Nonisothermal exactness is therefore a single compatibility condition posed on $(\s,T)$, not merely a family of isothermal exactness statements indexed by temperature.

Accordingly, the relevant object is a mobility-weighted capillary field on the
augmented saturation--temperature state domain. Theorem~\ref{thm:thermal_gtd}
makes this precise and shows how the classical isothermal compatibility
relations are embedded in the enlarged nonisothermal structure.

The mathematical structure used here is deliberately standard. The
global-pressure correction is exact when the mobility-weighted capillary field
is generated by a scalar potential. Equivalently, the line integral of the
capillary differential in Eq.~\eqref{eq:dPi_augmented} is path independent in
the augmented state domain $(\mathcal U)$. What is specific to the present
nonisothermal problem is the identification of the capillary field appropriate
to temperature-dependent multiphase capillary data and the resulting mixed
compatibility condition. The saturation-sector conditions reproduce the
classical isothermal TD/gTD constraints, while the mixed condition
$(\partial_T A_i=\partial_{s_i}B)$ is the additional nonisothermal obstruction.

\subsection{Slice-wise versus full exactness}

Suppose that temperature varies in space and time and that one seeks a scalar
state potential of the form
\begin{equation}
\Pi=\Pi(\s,T).
\label{eq:Pi_sT}
\end{equation}
By the chain rule, already stated in Eq.~\eqref{eq:chain_rule_Pi_intro}, its
spatial gradient contains both saturation-gradient and temperature-gradient
contributions. This shows why fixed-temperature exactness is not sufficient:
besides reproducing the mobility-weighted capillary contribution in the
saturation directions, a nonisothermal state potential must also reproduce the
thermal contribution proportional to \(\gradx T\).

It is therefore useful to distinguish two notions.
\emph{Slice-wise exactness} means exactness at fixed temperature: for each $T\in\cT$ there exists a scalar potential $\Pi_T(\s)$ such that
\begin{equation}
A_i(\s,T)=\partial_{s_i}\Pi_T(\s),
\qquad i=1,\dots,n_p-1,
\end{equation}
on the saturation simplex.
\emph{Full nonisothermal exactness} is stronger: it requires a \emph{single}
scalar potential \(\Pi(\s,T)\) on the augmented state domain such that
\begin{equation}
\mathbf C(\s,T)=\nabla_{(\s,T)}\Pi(\s,T),
\label{eq:full_exact_C_gradient}
\end{equation}
with the componentwise meaning already given in
Eq.~\eqref{eq:exactness_intro}.
The difference between these notions is the mixed saturation--temperature compatibility condition
\(
\partial_T A_i=\partial_{s_i}B
\),
in addition to the saturation-sector relations
\(
\partial_{s_j}A_i=\partial_{s_i}A_j
\).
This is the genuinely new obstruction introduced by the nonisothermal setting.

\subsection{Augmented-state mobility-weighted capillary field}

To formulate the problem geometrically, we use the mobility-weighted
capillary coefficients \(A_i(\s,T)\) and \(B(\s,T)\) introduced in
Eqs.~\eqref{eq:A_intro} and \eqref{eq:B_intro}. The coefficients \(A_i\)
describe the capillary response in the independent saturation directions,
whereas \(B\) gives the corresponding response in the temperature direction.
Together they define the augmented
capillary field
\begin{equation}
\mathbf C(\s,T)
=
\big(A_1(\s,T),\ldots,A_{n_p-1}(\s,T),B(\s,T)\big).
\label{eq:C_augmented}
\end{equation}
The associated formal capillary differential increment is
\begin{equation}
d\Pi
=
\sum_{i=1}^{n_p-1}A_i(\s,T)\,\mathrm d s_i
+
B(\s,T)\,\mathrm d T .
\label{eq:dPi_augmented}
\end{equation}
This is the nonisothermal analogue of the isothermal TD/gTD capillary
differential. Exactness means that this differential increment is generated by
a scalar potential \(\Pi(\s,T)\), or equivalently that
\begin{equation}
\nabla_{(\s,T)}\Pi(\s,T)=\mathbf C(\s,T).
\label{eq:gradPi_C_augmented}
\end{equation}
In practical terms, the capillary correction is then path independent in the
saturation--temperature state domain.

\begin{definition}[Nonisothermal gTD condition]
We say that the constitutive data
\((\{\lambda_\alpha\},\{p_{c,\alpha}\})\) satisfy the
\emph{nonisothermal gTD condition} on a subset
\(\mathcal U\subset\Delta\times\cT\) if the augmented capillary field
\(\mathbf C(\s,T)\) is conservative on \(\mathcal U\), equivalently if it is
locally the gradient of a scalar potential.
\label{def:thermal_gtd}
\end{definition}

Restricted to a constant-temperature slice, the augmented capillary field
reduces to the classical saturation-domain capillary field. Full nonisothermal
exactness requires more: the compatibility conditions must hold on the full
augmented state domain, so that the slice-wise potentials are mutually
compatible through the temperature direction.

\subsection{Exactness criterion on the augmented state domain}

We now state the main structural result.
The notation $\grady\Pi=\mathbf C$ is used in the Euclidean coordinates
$y=(s_1,\ldots,s_{n_p-1},T)$ on the state domain. Equivalently, the associated
differential is
\begin{equation}
    \mathrm d\Pi
=
\sum_{i=1}^{n_p-1}A_i\,\mathrm d s_i
+
B\,\mathrm dT .
\end{equation}
Thus no additional geometric structure beyond these state coordinates is
assumed.

\begin{theorem}[Nonisothermal gTD exactness criterion]
\label{thm:thermal_gtd}
Let
\begin{equation}
\mathcal U\subset \operatorname{int}\Delta\times (T_{\min},T_{\max})
\end{equation}
be an open regular subset of the augmented state domain. Assume that, for each
phase $\alpha$, the functions $\lambda_\alpha$ and $p_{c,\alpha}$ are of class
$C^2$ in $(\s,T)$ on $\mathcal U$, and that
\begin{equation}
\Lam(\s,T)=\sum_{\alpha=1}^{n_p}\lambda_\alpha(\s,T)>0,
\qquad
\forall (\s,T)\in\mathcal U.
\label{eq:Lambda_positive}
\end{equation}
Then the following statements are equivalent:

\begin{enumerate}[label=(\roman*)]
\item There exists, locally on $\mathcal U$, a scalar potential
$\Pi(\s,T)$ such that
\begin{align}
\partial_{s_i}\Pi(\s,T)&=A_i(\s,T),
\qquad i=1,\ldots,n_p-1,
\label{eq:Pi_system_s}
\\
\partial_T\Pi(\s,T)&=B(\s,T).
\label{eq:Pi_system_T}
\end{align}

\item The augmented capillary field
\begin{equation}
\mathbf C(\s,T)=\big(A_1,\ldots,A_{n_p-1},B\big)
\end{equation}
is locally conservative on $\mathcal U$, equivalently locally a gradient field.

\item The compatibility relations
\begin{align}
\partial_{s_j}A_i(\s,T)&=\partial_{s_i}A_j(\s,T),
\qquad
1\le i,j\le n_p-1,
\label{eq:compat_ss}
\\
\partial_T A_i(\s,T)&=\partial_{s_i}B(\s,T),
\qquad
1\le i\le n_p-1,
\label{eq:compat_sT}
\end{align}
hold throughout $\mathcal U$.
\end{enumerate}

If, in addition, $\mathcal U$ is path connected and simply connected, then
$\Pi$ exists globally on $\mathcal U$ and is unique up to an additive constant.
With
\begin{equation}
\pg = p_r + \Pi(\s,T),
\label{eq:pg_exact}
\end{equation}
the capillary contribution is represented exactly by a state potential, and
the corresponding global-pressure formulation is exactly equivalent to the
phase-pressure formulation at the continuous Darcy-model level.
\end{theorem}

Equations \eqref{eq:compat_ss} recover the familiar saturation-sector integrability conditions from the isothermal theory. Equations \eqref{eq:compat_sT} are new: they couple neighboring temperature slices and enforce that the thermal component of the mobility-weighted capillary structure can be integrated consistently together with the saturation component. Full nonisothermal exactness requires both sets of relations.

\begin{proof}
Using the reference-phase decomposition
\eqref{eq:palpha_pc_intro} and its spatial derivative
\eqref{eq:phase_grad_state_intro}, substitution into
\eqref{eq:pg_identity} gives

\begin{widetext}
\begin{align}
\Lam(\s,T)\,\gradx \pg
&=
\Lam(\s,T)\,\gradx p_r
+
\sum_{i=1}^{n_p-1}
\left(
\sum_{\alpha=1}^{n_p}\lambda_\alpha(\s,T)\,\partial_{s_i}p_{c,\alpha}(\s,T)
\right)\gradx s_i
+
\left(
\sum_{\alpha=1}^{n_p}\lambda_\alpha(\s,T)\,\partial_T p_{c,\alpha}(\s,T)
\right)\gradx T.
\label{eq:proof_1}
\end{align}

Using the definitions \eqref{eq:A_intro} and \eqref{eq:B_intro}, this becomes
\begin{align}
\Lam(\s,T)\,\gradx \pg
&=
\Lam(\s,T)\,\gradx p_r
+
\Lam(\s,T)\sum_{i=1}^{n_p-1}A_i(\s,T)\,\gradx s_i
+
\Lam(\s,T)B(\s,T)\,\gradx T.
\label{eq:proof_2}
\end{align}
\end{widetext}

Hence exact equivalence with the ansatz
\begin{equation}
\pg=p_r+\Pi(\s,T)
\label{eq:proof_pg}
\end{equation}
holds if and only if the componentwise conditions
\eqref{eq:exactness_intro} hold, or equivalently
\(\grady\Pi=\mathbf C\).

The relations in item (i) state precisely that the augmented capillary field
\(\mathbf C=(A_1,\ldots,A_{n_p-1},B)\) is the gradient of the scalar potential
\(\Pi\). Therefore its mixed partial derivatives must be mutually compatible,
which gives \eqref{eq:compat_ss}--\eqref{eq:compat_sT}. Conversely, if these
compatibility relations hold on a sufficiently small state-domain neighborhood,
the line integral of \(\mathbf C\) is path independent and defines a local
scalar potential. If \(\mathcal U\) is simply connected, the same construction
gives a global potential on \(\mathcal U\), unique up to an additive constant.
This proves the equivalence of (i)--(iii), and therefore the exactness claim.
\end{proof}

The criterion is applied on regular open subsets of the admissible
saturation--temperature state domain, where the mobilities and
capillary-pressure derivatives are smooth and the total mobility satisfies
$\Lambda_t>0$. This is the natural domain for the mixed-partial compatibility argument.
At the simplex boundary, where one or more
phase saturations approach zero, endpoint relative permeabilities may vanish
and capillary-pressure derivatives may become singular or degenerate. These
boundary regimes require the usual phase-appearance restrictions,
residual-saturation cutoffs, or regularized capillary laws. In the numerical
examples, the diagnostics are therefore evaluated on an interior regularized
simplex, rather than directly on degenerate saturation edges.

\subsection{Interpretation, isothermal limit, and trajectory-wise exactness}

Theorem~\ref{thm:thermal_gtd} separates two levels of structure. At the
algebraic level, the mobility-weighted capillary contribution can always be
written as the augmented field \(\mathbf C(\s,T)\). At the differential level,
it corresponds to a genuine scalar potential only when \(\mathbf C\) is a
conservative field. This is the precise sense in which global pressure is an
exact reduction rather than a convenient approximation.
If \(T\) is fixed, then \(\mathrm d T=0\), and the augmented capillary field
reduces to the saturation-domain capillary field. The mixed conditions
\eqref{eq:compat_sT} disappear, and Theorem~\ref{thm:thermal_gtd} reduces
exactly to the usual isothermal TD/gTD criterion
\cite{Chavent2009,DiChiara2010}.

Taken together, these remarks show that the nonisothermal criterion is a strict extension of TD/gTD: it retains the saturation-sector conditions and adds the mixed relations
\(
\partial_T A_i=\partial_{s_i}B
\),
which enforce compatibility between the slice-wise potentials as \(T\) varies. Exactness can therefore hold on every fixed-temperature slice and still fail for genuinely nonisothermal evolutions.

\begin{corollary}[Exactness along a thermodynamic trajectory]
\label{cor:trajectory}
Let \(\Theta\subset\mathbb R\) be an open temperature interval and let
\(\gamma:[0,t_f]\to \Delta\times\Theta\) be a smooth thermodynamic path
\(t\mapsto(\s(t),T(t))\). Let \(\mathbf y=(\s,T)\).
If the augmented capillary field \(\mathbf C\) is conservative on an open
neighborhood of \(\gamma([0,t_f])\), so that
\(\mathbf C=\nabla_{\mathbf y}\Pi\) there, then for every \(t\in[0,t_f]\),
\begin{equation}
\Pi(\gamma(t))-\Pi(\gamma(0))
=
\int_{\gamma|_{[0,t]}}\mathbf C\cdot \mathrm d\mathbf y .
\label{eq:trajectory_integral}
\end{equation}
For a generic path \(\gamma(\tau)=(\s(\tau),T(\tau))\), the line integral is
\begin{equation}
\int_\gamma \mathbf C\cdot \mathrm d\mathbf y
:=
\int_0^1
\mathbf C\big(\s(\tau),T(\tau)\big)
\cdot
\big(\dot{\s}(\tau),\dot T(\tau)\big)
\,\mathrm d\tau .
\label{eq:path_integral_explicit}
\end{equation}
Equivalently, the line integral depends only on the endpoints of the
restricted path segment \(\gamma|_{[0,t]}\) and not on its parametrization.
If global exactness fails on \(\Delta\times\Theta\) but holds on a simply
connected neighborhood of \(\gamma([0,t_f])\), then the global-pressure
reduction remains exact along that trajectory.
\end{corollary}

This corollary is especially useful in applications, where a simulation
typically explores only a restricted portion of the full state domain. Thus, even if nonisothermal exactness fails globally, the reduction may remain exact on the dynamically visited region. Conversely, a moving thermal front may push the solution trajectory across the boundary between exact and nonexact regions.

\subsection{Defect diagnostics}
\label{sec:defect_diagnostics}

Theorem~\ref{thm:thermal_gtd} provides an exact criterion for
global-pressure equivalence. In applications, however, the functions entering
this criterion are not usually known analytically from first principles. They
are prescribed as closure laws, fitted correlations, or tabulated data for the
phase mobilities $(\lambda_\alpha(\s,T))$ and capillary pressures
$(p_{c,\alpha}(\s,T))$. The resulting solution also explores only a subset of
the full saturation--temperature state domain. We therefore introduce
operational defect measures that identify where exactness holds or fails,
separate saturation-sector nonintegrability from genuinely nonisothermal
coupling, and provide scalar diagnostics that can be tracked along solution
trajectories.

\medskip

\noindent\textbf{(i) Saturation-sector defect.}
The first component measures failure of saturation-sector integrability,
\begin{equation}
\mathcal D_{ss}(\s,T)
=
\max_{i,j}
\left|
\partial_{s_j}A_i(\s,T)-\partial_{s_i}A_j(\s,T)
\right|.
\label{eq:defect_ss}
\end{equation}
For fixed $T$, $\mathcal D_{ss}$ reduces to the classical TD/gTD integrability defect on the saturation simplex.

\medskip

\noindent\textbf{(ii) Mixed saturation--temperature defect.}
The second component measures failure of the mixed compatibility relations,
\begin{equation}
\mathcal D_{sT}(\s,T)
=
\max_i
\left|
\partial_T A_i(\s,T)-\partial_{s_i}B(\s,T)
\right|.
\label{eq:defect_sT}
\end{equation}
This quantity has no isothermal counterpart: it vanishes identically when $T$ is frozen. It measures the specific obstruction to assembling the slice-wise potentials $\Pi_T(\s)$ into a single augmented-state potential $\Pi(\s,T)$.

\medskip

\noindent\textbf{(iii) Combined defect and usage.}
A convenient scalar indicator of nonisothermal gTD failure is
\begin{equation}
\mathcal D_{\mathrm{gTD}}(\s,T)
=
\max\{\mathcal D_{ss}(\s,T),\,\mathcal D_{sT}(\s,T)\}.
\label{eq:full_defect}
\end{equation}
In what follows, $\mathcal D_{\mathrm{gTD}}$ is used in three ways:
\begin{enumerate}
\item \emph{State-domain maps:} by plotting $\mathcal D_{ss}$ and $\mathcal D_{sT}$ on $\Delta$ across temperatures, we identify regions where global-pressure exactness is expected to hold and where it must fail.
\item \emph{Trajectory diagnostics:} evaluating $\mathcal D_{\mathrm{gTD}}(\s(x,t),T(x,t))$ along simulated trajectories quantifies when a moving thermal front drives the system into nonexact regimes and whether the loss of exactness is primarily saturation-driven or thermal-coupling-driven.
\item \emph{Model-control input:} in the numerical section we use $\mathcal D_{\mathrm{gTD}}$ to classify regimes (exact / weakly nonexact / strongly nonexact) and to interpret the behavior of the conservative surrogate introduced later.
\end{enumerate}

\noindent
These diagnostics do not replace Theorem~\ref{thm:thermal_gtd}; they provide practical summaries of its compatibility conditions on the subset of $(\s,T)$ actually explored by a model or simulation.

\section{Minimal nonisothermal matrix--fracture system with evolving aperture}
\label{sec:matrix_fracture}

Theorem~\ref{thm:thermal_gtd} is a structural result: it characterizes when a nonisothermal global-pressure representation is exactly equivalent to the phase-pressure formulation on the augmented state domain. To assess the practical relevance of that result, we embed it in a minimal conservative matrix--fracture model in which saturation and temperature evolve and fracture transmissivity changes under thermo-hydraulic forcing. The goal is not to construct a full geothermal simulator, but to introduce the smallest reduced system that contains the ingredients needed to connect augmented-state exactness to fractured nonisothermal flow. Reduced fluid-and-heat models of this type are standard in fractured-media and geothermal applications \cite{PruessNarasimhan1985,Wang2021FracturedGeothermal}.

The model contains four coupled parts: a pressure subsystem determining the total flux in the matrix and in the fracture, a saturation subsystem describing advective--capillary phase redistribution, a thermal subsystem accounting for advection, conduction, and matrix--fracture heat exchange, and an aperture subsystem through which pressure and temperature modify fracture transmissivity. Each part is retained only because it is needed to move the local state through the augmented domain $(\s,T)$ and to test whether the exact global-pressure structure is preserved or lost.

Two simplifying choices are made from the outset. First, the fracture is
represented as a lower-dimensional domain embedded in the porous matrix. This
is the standard mixed-dimensional reduction when the fracture aperture is small
relative to the matrix length scale, while the fracture remains hydraulically
important \cite{Berre2019,Ahmed2017,Wang2021FracturedGeothermal}. Second, the
thermal model is written under local thermal equilibrium (LTE), so that each
continuum carries a single temperature variable. This keeps the state
description compact and isolates the role of the augmented state domain. LTE is
used here as a reduced closure, not as a universal assumption: local thermal
nonequilibrium can become important in heterogeneous porous media,
fracture-dominated systems, high-velocity flow, systems with large thermal
property contrasts or interfacial resistance, and near sharp thermal fronts
\cite{Hamidi2019,Heinze2024,Kostelecky2026LTNE}.

\subsection{Geometry and unknowns}

Let $\Omega\subset\mathbb R^d$ denote the porous matrix domain and let $\Gamma\subset\Omega$ denote a lower-dimensional fracture domain. 
Here \(\nabla_{\Gamma}\) and \(\nabla_{\Gamma}\!\cdot\) denote,
respectively, the tangential gradient and tangential divergence along the
fracture manifold \(\Gamma\). For a scalar field \(u\) defined on
\(\Gamma\), \(\nabla_{\Gamma}u\) is the projection of the full spatial
gradient onto the tangent plane of the fracture; correspondingly,
\(\nabla_{\Gamma}\!\cdot\mathbf q_\Gamma\) is the divergence of a tangential
flux \(\mathbf q_\Gamma\) along \(\Gamma\).
Superscripts $m$ and $f$ refer to matrix and fracture quantities, respectively. The minimal unknown set is
\[
(\pg^m,\pg^f),\quad
(\s^m,\s^f),\quad
(T^m,T^f),\quad
b=b(x,t)\ \text{on }\Gamma,
\]
where $b$ is the hydraulic aperture.

The variables $\pg^m$ and $\pg^f$ are the reduced pressure variables driving total flow in the two continua. The saturation vectors $\s^m$ and $\s^f$ describe the local phase configuration, $T^m$ and $T^f$ determine the thermal state entering mobilities, densities, and capillary laws, and $b$ is the internal variable through which thermo-hydraulic forcing modifies fracture transmissivity. This is the smallest state vector needed to couple nonisothermal constitutive evolution to matrix--fracture transport and aperture-controlled conductivity.

\subsection{Matrix and fracture Darcy laws}

In the matrix, phase velocities obey
\begin{equation}
\valpha^{m}
=
-\K\,
\lam_\alpha^{m}(\s^m,T^m)
\Big(
\gradx p_\alpha^{m}-\rho_\alpha(T^m)\,g\,\gradx z
\Big),
\label{eq:darcy_matrix}
\end{equation}
Here $\K$ is the matrix permeability tensor. On the fracture domain the phase velocities obey the Darcy law along the fracture
\begin{equation}
\valpha^{f}
=
-\mathcal T_f(b)\,
\lam_\alpha^{f}(\s^f,T^f)
\Big(
\gradG p_\alpha^{f}-\rho_\alpha(T^f)\,g\,\gradG z
\Big).
\label{eq:darcy_fracture}
\end{equation}
Here $\gradG$ denotes the gradient along the fracture, and
$\mathcal T_f(b)$ is the fracture transmissivity. The total mobilities are
\begin{equation}
\Lam^\kappa(\s^\kappa,T^\kappa)=\sum_{\alpha=1}^{n_p}\lam_\alpha^\kappa(\s^\kappa,T^\kappa),
\qquad
\kappa\in\{m,f\}.
\label{eq:total_mobs_kappa}
\end{equation}

The matrix is treated as a full-dimensional porous continuum, whereas the fracture acts as a lower-dimensional conductive pathway carrying flow along the fracture. Because the fracture is typically more conductive than the matrix, injection and cooling may drive the fracture and matrix across different regions of the augmented state domain \cite{Berre2019,Wang2021FracturedGeothermal}.

The fracture transmissivity is taken in cubic-law form,
\begin{equation}
\mathcal T_f(b)=\frac{b^3}{12},
\label{eq:cubic_law}
\end{equation}
with roughness and shape corrections absorbed into $b$ or, if needed, into a multiplicative factor. Here $b=b(x,t)$ denotes the \emph{hydraulic aperture} on $\Gamma$, namely the effective aperture parameter yielding the same laminar throughflow as the actual rough-walled fracture under the reduced Darcy law along the fracture. Equation~\eqref{eq:cubic_law} is the classical baseline constitutive relation for aperture-controlled fracture conductivity \cite{Witherspoon1980}. Its role here is structural: because $\mathcal T_f$ depends cubically on $b$, even modest aperture changes can induce large conductivity changes and therefore large shifts in the dynamically sampled state trajectory $(\s,T)$.

\subsection{Global pressures in matrix and fracture}

The global-pressure construction is introduced separately in the matrix and in the fracture:
\begin{align}
\sum_{\alpha=1}^{n_p}
\lam_\alpha^m(\s^m,T^m)\,\gradx p_\alpha^m
&=
\Lam^m(\s^m,T^m)\,\gradx \pg^m,
\label{eq:pg_matrix}
\\
\sum_{\alpha=1}^{n_p}
\lam_\alpha^f(\s^f,T^f)\,\gradG p_\alpha^f
&=
\Lam^f(\s^f,T^f)\,\gradG \pg^f.
\label{eq:pg_fracture}
\end{align}
These identities do not add new physics; they reorganize the multiphase pressure coupling into a reduced form. Whenever the nonisothermal gTD condition of Theorem~\ref{thm:thermal_gtd} holds continuum-wise, the corresponding global pressures can be written as state potentials,
\begin{equation}
\pg^\kappa = p_r^\kappa + \Pi^\kappa(\s^\kappa,T^\kappa),
\qquad
\kappa\in\{m,f\}.
\label{eq:pg_kappa_exact}
\end{equation}

Treating the matrix and the fracture separately is essential. Because the fracture responds more rapidly to cooling and injection, the fracture state may enter a nonexact region of $(\s,T)$ while the matrix remains in an exact one, or conversely. The continuum-by-continuum formulation is therefore the natural setting in which exactness can be lost or recovered differently in the two continua.

\subsection{Pressure equations}

At the reduced level, the total Darcy laws are
\begin{align}
\mathbf v_t^{m}
&=
-\mathbf K\,\Lambda_t^m(\mathbf s^m,T^m)
\left(
\nabla p_g^m
-
\rho_t^m(\mathbf s^m,T^m)\,g\,\nabla z
\right),
\label{eq:vt_m}
\\
\mathbf v_t^{f}
&=
-\mathcal T_f(b)\,\Lambda_t^f(\mathbf s^f,T^f)
\left(
\gradG p_g^f
-
\rho_t^f(\mathbf s^f,T^f)\,g\,\gradG z
\right).
\label{eq:vt_f}
\end{align}
where the mobility-weighted total densities are defined by
\begin{equation}
\rho_t^\kappa(\s^\kappa,T^\kappa)
=
\frac{1}{\Lam^\kappa(\s^\kappa,T^\kappa)}
\sum_{\alpha=1}^{n_p}
\lam_\alpha^\kappa(\s^\kappa,T^\kappa)\,\rho_\alpha(T^\kappa),
\label{eq:rho_total_weighted}
\end{equation}
for $\kappa\in\{m,f\}$. These are supplemented by the total mass balances
\begin{align}
\divv \vt^m &= q_t^m,
\label{eq:mass_total_m}
\\
\divG \vt^f &= q_t^f + \cQ_{mf},
\label{eq:mass_total_f}
\end{align}
where $\mathcal Q_{mf}$ denotes the conservative matrix--fracture volumetric
exchange written explicitly in the lower-dimensional fracture balance. The
opposite contribution is enforced on the matrix side through the corresponding
interface flux condition.

Equations~\eqref{eq:vt_m}--\eqref{eq:mass_total_f} define the pressure
subsystem of the reduced model. They are used here as a quasi-steady
incompressible total-pressure closure: for each saturation, temperature, and
aperture field, the pressure solve determines the corresponding instantaneous
total fluxes in the matrix and fracture continua. Total compressive storage is
therefore not included in this pressure subsystem. Aperture-dependent storage,
however, is retained explicitly in the conservative fracture saturation and
energy balances. A fully compressible extension would add the corresponding
total-storage terms, including aperture-change contributions, to the pressure
equation.

At this level, the formulation is neutral as to whether the global pressure is
an exact state potential or a projected surrogate. That distinction enters
through the constitutive treatment of the mobility-weighted capillary
contribution.

\subsection{Saturation transport}

For $i=1,\dots,n_p-1$, let
\begin{equation}
f_i^\kappa=\frac{\lam_i^\kappa}{\Lam^\kappa},
\qquad
\kappa\in\{m,f\},
\label{eq:fractional_flow_kappa}
\end{equation}
denote the continuum-wise fractional-flow coefficients.

Because the capillary pressures depend on both saturation and temperature, the
relative phase flux contains both saturation-capillary and thermal-capillary
redistribution terms. In the matrix, we write
\begin{align}
\mathbf j_{i,\mathrm{rel}}^m
&=
-\K\,\lam_i^m
\sum_{j=1}^{n_p-1}
\left(
\partial_{s_j}p_{c,i}^m-A_j^m
\right)
\gradx s_j^m
\nonumber \\
&\quad
-\K,\lam_i^m
\left(
\partial_T p_{c,i}^m-B^m
\right)
\gradx T^m
\nonumber \\
&\quad
+\K,\lam_i^m
\left(
\rho_i-\rho_t^m
\right)
g\,\gradx z .
\label{eq:jrel_matrix}
\end{align}
On the fracture, the corresponding lower-dimensional relative flux is
\begin{align}
\mathbf j_{i,\mathrm{rel}}^f
&=
-\mathcal T_f(b)\,\lam_i^f
\sum_{j=1}^{n_p-1}
\left(
\partial_{s_j}p_{c,i}^f-A_j^f
\right)
\gradG s_j^f
\nonumber \\
&\quad
-\mathcal T_f(b),\lam_i^f
\left(
\partial_T p_{c,i}^f-B^f
\right)
\gradG T^f
\nonumber \\
&\quad
+\mathcal T_f(b),\lam_i^f
\left(
\rho_i-\rho_t^f
\right)
g\,\gradG z .
\label{eq:jrel_fracture}
\end{align}
The first terms in Eqs.~\eqref{eq:jrel_matrix} and
\eqref{eq:jrel_fracture} are the usual saturation-capillary redistribution
terms. The second terms are the thermal-capillary redistribution terms induced
by the temperature dependence of the capillary pressures. The last terms are
gravity-segregation contributions. In horizontal reduced benchmarks, or when
density contrasts are neglected, the gravity terms are set to zero.

The conservative saturation balances are
\begin{align}
\partial_t\left(\poro^m s_i^m\right)
+
\divv\left(
f_i^m(\s^m,T^m)\,\vt^m
+
\mathbf j_{i,\mathrm{rel}}^m
\right)
&=
q_i^m+\Tmf_i,
\label{eq:sat_m}
\\
\partial_t\left(b\,\poro^f s_i^f\right)
+
\divG\left(
f_i^f(\s^f,T^f)\,\vt^f
+
\mathbf j_{i,\mathrm{rel}}^f
\right)
&=
q_i^f-\Tmf_i.
\label{eq:sat_f}
\end{align}
Here $\Tmf_i$ denotes the conservative matrix--fracture exchange term for the
$i$th independent saturation variable. The sign convention is such that
$\Tmf_i>0$ transfers that variable from the fracture to the matrix.

Equations~\eqref{eq:sat_m}--\eqref{eq:sat_f} have a two-part transport
structure. The first part is advection by the total flux through the
fractional-flow coefficients. The second part is the relative redistribution
flux $\mathbf j_{i,\mathrm{rel}}^\kappa$, which contains saturation-capillary
redistribution, thermal-capillary redistribution, and, when retained,
gravitational segregation. When nonisothermal gTD exactness holds, the
capillary part of this redistribution is consistent with the scalar state
potential $\Pi^\kappa(\s^\kappa,T^\kappa)$. When exactness fails, the balance
laws remain conservative, but the capillary contribution is interpreted through
the projection-based surrogate introduced in Sec.~\ref{sec:projection}.

The fracture transport equation differs from its matrix counterpart because
transport occurs along the fracture and because the fracture storage is
aperture-integrated. Here $\poro^m$ and $\poro^f$ are taken as constant
porosities in the matrix and fracture continua, respectively. Thus the
difference between Eqs.~\eqref{eq:sat_m} and \eqref{eq:sat_f} is not a
time-dependent fracture porosity, but the evolving aperture factor in the
reduced fracture storage. The matrix storage is $\poro^m s_i^m$, whereas the
fracture storage is $b\poro^f s_i^f$. For constant $\poro^f$, the
aperture-integrated fracture storage satisfies
\begin{equation}
\partial_t\left(b\,\poro^f s_i^f\right)
=
b\,\poro^f\,\partial_t s_i^f
+
\poro^f s_i^f\,\partial_t b .
\label{eq:aperture_storage_expansion}
\end{equation}
Thus aperture evolution affects the fracture balance through the storage factor
$b\poro^f s_i^f$, while $\poro^f$ itself is kept constant.

\subsection{Energy equations and matrix--fracture heat exchange}

We adopt local thermal equilibrium (LTE) as the minimal thermal closure. Under this assumption, the solid skeleton and the fluid phases within each continuum share a single temperature, so the energy balance is written with one thermal state variable in the matrix and one on the fracture. This choice is made to isolate the effect of augmented-state exactness rather than to resolve subcontinuum thermal lag. It should nevertheless be read as a reduced closure: LTNE effects may become relevant in heterogeneous media, in fracture-dominated systems, and near sharp thermal fronts \cite{Hamidi2019,Heinze2024}.

A minimal matrix energy balance is
\begin{widetext}
\begin{equation}
\partial_t U^m(\s^m,T^m)
+
\divv\!\left(
\sum_{\alpha=1}^{n_p}
h_\alpha(T^m)\rho_\alpha(T^m)\,\valpha^m
-
k_T^m\,\gradx T^m
\right)
=
Q_T^m+\cH_{mf},
\label{eq:energy_m}
\end{equation}
while the fracture energy balance is written in aperture-integrated form,
\begin{equation}
\partial_t\!\big(
b\,U^f(\s^f,T^f)
\big)
+
\divG\!\left(
\sum_{\alpha=1}^{n_p}
h_\alpha(T^f)\rho_\alpha(T^f)\,\valpha^f
-
k_T^f\,\gradG T^f
\right)
=
Q_T^f-\cH_{mf}.
\label{eq:energy_f}
\end{equation}
\end{widetext}

Here $U^\kappa$ denotes the continuum internal-energy density, $h_\alpha$ the phase enthalpies, and $k_T^\kappa$ the effective thermal conductivities. In the fracture equation, \(k_T^f\) is understood as an aperture-integrated effective thermal conductivity along the fracture on the reduced fracture domain. The terms $\cH_{mf}$, like $\cQ_{mf}$ and $\Tmf_i$, are conservative exchange closures between matrix and fracture; only the thermal exchange term is specialized below.

These balances combine advection of enthalpy by the phase fluxes, thermal conduction within each continuum, and heat exchange between matrix and fracture. A simple reduced closure for the latter is Newton cooling,
\begin{equation}
\cH_{mf}=\eta_h\big(T^f - T^m|_\Gamma\big),
\label{eq:newton_cooling}
\end{equation}
with $\eta_h$ an inter-continuum heat-transfer coefficient. With this sign convention, \(\cH_{mf}>0\) corresponds to net heat transfer from fracture to matrix. More elaborate aperture- and flow-dependent heat-transfer coefficients are also available \cite{Heinze2017HeatCoeff}. For the present purpose, however, \eqref{eq:newton_cooling} is sufficient to retain the key coupling: fracture and matrix temperatures do not evolve independently, so cooling in the fracture perturbs the local exactness condition in both continua.

The LTE assumption is a closure choice, not a universal property of fractured
thermal flow. The LTE closure is appropriate when interphase and solid--fluid heat exchange
are rapid compared with advective thermal-front propagation. Its accuracy may
decrease near sharp thermal fronts, in high-velocity fractures, in strongly
heterogeneous matrix blocks, in systems with large thermal-property contrasts
or interfacial thermal resistance, or when fracture and matrix temperatures
exhibit appreciable lag \cite{Kostelecky2026LTNE}.
In such regimes, a local thermal
nonequilibrium model with separate fluid, solid, matrix, and/or fracture
temperatures would be more appropriate. The exactness theory would then extend
to a larger augmented state domain containing these additional temperature
variables, with additional mixed compatibility conditions.

\subsection{Reduced thermo-hydro-mechanical aperture law}

The final ingredient is the evolution of fracture aperture. In geothermal settings, fracture opening and closure are affected by both hydraulic pressure and thermal stress. Thermo-poroelastic analyses of heat extraction show that cooling, leak-off, and stress redistribution can alter fracture aperture appreciably, thereby modifying fracture conductivity and heat-transfer performance \cite{Ghassemi2008,Pandey2017THM}. To capture this feedback without introducing a full contact-mechanics model, we use a reduced thermo-hydro-mechanical closure.

We write the effective normal stress as
\begin{equation}
\sigma_n'=\sigma_n-\alpha_B\,\bar p_f,
\label{eq:effective_stress}
\end{equation}
where $\sigma_n$ is the imposed normal stress, $\alpha_B$ is a Biot-type coefficient, and $\bar p_f$ is a local representative fracture pressure entering the reduced stress law, for instance \(p_r^f\), \(\pg^f\), or a suitable local average. A thermoelastic stress shift is modeled by
\begin{equation}
\sigma_{th}(T^f)=K_{th}(T^f-T_0).
\label{eq:thermal_shift}
\end{equation}
We then prescribe the reduced aperture evolution law
\begin{equation}
\partial_t b
=
\frac{1}{\tau_b}
\Big(
b_{\mathrm{eq}}(\sigma_n'+\sigma_{th}(T^f)) - b
\Big)
+
\eta_q\,\norm{\vt^f}^2.
\label{eq:b_evolution}
\end{equation}

The first term in \eqref{eq:b_evolution} describes relaxation toward a stress-controlled equilibrium aperture. The second term is introduced here as a reduced rate-dependent dilation surrogate feeding back from fracture flow intensity to transmissivity. It is not intended as a detailed friction or contact law.

A convenient monotone stress--aperture closure is
\begin{equation}
b_{\mathrm{eq}}(\Sigma)
=
b_{\mathrm{res}}
+
\frac{b_0-b_{\mathrm{res}}}{1+\Sigma/\Sigma_*},
\qquad
\Sigma=\sigma_n'+\sigma_{th}(T^f),
\label{eq:bb_like}
\end{equation}
which reproduces the qualitative monotone closure behavior often used in reduced stress-sensitive fracture-flow models \cite{Bisdom2016}.
Together, \eqref{eq:cubic_law}, \eqref{eq:effective_stress},
\eqref{eq:thermal_shift}, and \eqref{eq:b_evolution} close the reduced
thermo-hydro-mechanical feedback loop. The variables involved are the fracture
temperature $T^f$, the thermally induced normal stress
$\sigma_{th}(T^f)$, the effective normal stress $\sigma_n'$, the hydraulic
aperture $b$, the fracture transmissivity $\mathcal T_f(b)$, the fracture
total Darcy velocity $\vt^f$, and the representative fracture pressure
$\bar p_f$:
\begin{equation}
T^f
\leftrightarrow
\sigma_{th}(T^f)
\leftrightarrow
\sigma_n'
\leftrightarrow
b
\leftrightarrow
\mathcal T_f(b)
\leftrightarrow
\vt^f
\leftrightarrow
\bar p_f .
\end{equation}

This loop is the application-side reason the problem is dynamically interesting. Thermal perturbations modify stress and aperture; aperture modifies transmissivity; transmissivity modifies fracture flux and hence the saturation and temperature trajectories; and those trajectories determine whether the local nonisothermal gTD condition is satisfied. Aperture evolution is therefore not an ancillary constitutive detail, but one of the mechanisms by which the solution may move between exact and nonexact global-pressure regimes.
The aperture law above is intentionally reduced. It is not meant to replace full thermo-hydro-mechanical fracture models, rough-walled contact formulations, or frictional slip analyses. Its role is narrower: it provides the simplest feedback by which thermal forcing and fracture flow modify transmissivity and thereby move the local state through regions where the nonisothermal gTD condition holds or fails.

\section{Projection-based conservative closure beyond exactness}
\label{sec:projection}

Sections~\ref{sec:thermal_gtd} and \ref{sec:matrix_fracture} identify two qualitatively distinct regimes. When the nonisothermal gTD condition holds, the global-pressure reduction is exact. When the evolving state $(\s(x,t),T(x,t))$ enters a region of the augmented state domain where the compatibility relations of Theorem~\ref{thm:thermal_gtd} fail, no scalar potential $\Pi(\s,T)$ exists such that $\mathbf C=\nabla_{(\s,T)}\Pi$, and exact equivalence with the phase-pressure formulation is lost. In that regime, one may either revert to the phase-pressure formulation or retain the global-pressure viewpoint through a reduced closure. We adopt the second option and seek a surrogate that preserves the conservative structure of the reduced balance laws while making the loss of exactness explicit.

The key observation is that nonexactness has two distinct components. One is the failure of integrability within each fixed-temperature saturation slice. The other is the genuinely nonisothermal incompatibility encoded by the mixed conditions
\(
\partial_T A_i=\partial_{s_i}B
\).
The closure introduced below addresses only the first component: it replaces the saturation-sector field by its best gradient approximation on each isothermal slice, while leaving the mixed nonisothermal defect explicit. In this way, the reduced model remains conservative on each slice without claiming restoration of full nonisothermal exactness \cite{Chavent2009,tantardini2025globalbuckleyleverettmulticomponentflow}.

\subsection{Slice-wise projection at fixed temperature}

When the nonisothermal gTD condition fails in a continuum
$(\kappa\in\{m,f\})$, no exact state potential
$(\Pi^\kappa(\s^\kappa,T^\kappa))$ exists on the corresponding part of the
augmented state domain. A reduced construction can nevertheless be introduced
by fixing the temperature and projecting only the saturation-sector capillary
field onto gradient fields on the simplex. This projection is performed
separately in the matrix and in the fracture, so the continuum index
$(\kappa)$ is kept explicit.

For fixed $(\kappa\in\{m,f\})$ and fixed $(T)$, define the saturation-sector
capillary field
\begin{equation}
\bm C_s^\kappa(\s;T)
:=
\big(
A_1^\kappa(\s,T),\ldots,A_{n_p-1}^\kappa(\s,T)
\big).
\label{eq:Cs_kappa_projection}
\end{equation}
The functions $(A_i^\kappa)$ are the saturation components of the
mobility-weighted capillary field in continuum $(\kappa)$; they are not scalar
potentials. The projected scalar potential is denoted below by
$(\Pi^{\kappa,\star})$.

For each fixed \(T\), the projected potential
\(\Pi^{\kappa,\star}(\cdot;T)\in H^1(\Delta)\) is defined as the minimizer
of the least-squares functional
\begin{widetext}
\begin{equation}
\Pi^{\kappa,\star}(\cdot;T)
=
\arg\min_{\substack{\Pi\in H^1(\Delta)\\
\int_\Delta \Pi(\s;T)\,\dd\s=0}}
\frac12
\int_\Delta
\left|
\grads \Pi(\s;T)
-
\bm C_s^\kappa(\s;T)
\right|^2
\,\dd\s .
\label{eq:H1proj}
\end{equation}
\end{widetext}
The zero-mean constraint fixes the additive constant of the potential.

The constrained minimizer is characterized by vanishing first variation. Thus,
for every admissible test function \(\psi\in H^1(\Delta)\) satisfying
\(\int_\Delta\psi\,\dd\s=0\),
\begin{equation}
\int_\Delta
\left(
\grads\Pi^{\kappa,\star}
-
\bm C_s^\kappa
\right)
\cdot
\grads\psi\,
\dd\s
=0 .
\label{eq:weak_EL}
\end{equation}
Equation~\eqref{eq:weak_EL} is the weak Euler--Lagrange condition associated
with the constrained minimization problem \eqref{eq:H1proj}. Integrating it by
parts yields the equivalent Poisson--Neumann problem
\begin{align}
\Delta_s\Pi^{\kappa,\star}(\s;T)
&=
\grads\!\cdot \bm C_s^\kappa(\s;T)
\qquad \text{in } \Delta,
\label{eq:poisson_proj_in}
\\
\partial_{\bm n}\Pi^{\kappa,\star}(\s;T)
&=
\bm C_s^\kappa(\s;T)\cdot\bm n
\qquad \text{on } \partial\Delta,
\label{eq:poisson_proj_bd}
\end{align}
supplemented by the zero-mean constraint included in
Eq.~\eqref{eq:H1proj}, which fixes the additive constant. The Neumann compatibility condition is automatically
satisfied by the divergence theorem, so the zero-mean gauge is the only
additional normalization required.

Define the slice-wise residual in continuum \(\kappa\) by
\begin{equation}
\bm R^\kappa(\s;T)
=
\bm C_s^\kappa(\s;T)
-
\grads\Pi^{\kappa,\star}(\s;T).
\label{eq:Rfield}
\end{equation}
Then \(\bm R^\kappa\) is the non-gradient part of the saturation-sector field
on the fixed-temperature slice. In particular, \(\bm R^\kappa\equiv0\) if and
only if the saturation-sector field is already integrable at that temperature.

\begin{definition}[Fiberwise projected global pressure]
For each fixed temperature \(T\) and each continuum
\(\kappa\in\{m,f\}\), define
\begin{equation}
\pg^{\kappa,\star}
=
p_r^\kappa+\Pi^{\kappa,\star}(\s^\kappa;T^\kappa).
\label{eq:pg_star}
\end{equation}
We refer to \(\pg^{\kappa,\star}\) as the \emph{fiberwise projected global
pressure}.
\end{definition}

The adjective ``fiberwise'' is essential: the projection is performed
independently on each isothermal slice of the augmented state domain. It
therefore removes only the non-gradient component of the saturation-sector
field. It does not modify the mixed defect
\(\partial_T A_i-\partial_{s_i}B\), and hence does not restore full
nonisothermal exactness when that defect is nonzero.

Because the projection modifies only the constitutive representation of the
capillary contribution, the reduced balance laws in
Sec.~\ref{sec:matrix_fracture} remain in conservative form. The projected
global-pressure formulation is therefore a conservative reduced surrogate, not
an exactly equivalent phase-pressure formulation in mixed-nonexact regimes.

\subsection{Relation to the local nonisothermal defect}

The pointwise quantities $\cD_{ss}(\s,T)$ and $\cD_{sT}(\s,T)$ introduced in Sec.~\ref{sec:defect_diagnostics} diagnose local failure of exactness on the augmented state domain. The slice-wise projection residual \(\bm R^\kappa(\s;T)\) plays a complementary role. It is the \(L^2\)-optimal non-gradient remainder of the saturation-sector field at fixed temperature, whereas \(\cD_{ss}\) and \(\cD_{sT}\) measure pointwise violations of the compatibility relations.

This distinction is important. A small projection residual on a given temperature slice indicates that the saturation-sector field is close, in an integrated sense, to a gradient field on that slice. By contrast, a small local curl defect \(\cD_{ss}\) is a pointwise statement. The two diagnostics are therefore related but not interchangeable.

Equally important, a small saturation-sector residual does \emph{not} imply near-exactness of the full nonisothermal problem. It only indicates that the isothermal part of the obstruction is weak. The genuinely nonisothermal obstruction is controlled separately by the mixed compatibility defect involving \(B(\s,T)\), namely
\(
\partial_TA_i-\partial_{s_i}B
\).
For this reason, the projected closure must always be interpreted together with the mixed diagnostics.

\subsection{Nonexactness quantifiers}

A practical scalar measure of the saturation-sector projection defect at fixed
temperature is
\begin{equation}
\mathfrak D^\kappa(T)
=
\left(
\int_\Delta |\bm R^\kappa(\s;T)|^2\,\dd\s
\right)^{1/2}.
\label{eq:Dproj}
\end{equation}
This quantity measures how far the slice-wise field in continuum \(\kappa\)
is from being representable by a scalar saturation potential on the simplex.

The genuinely nonisothermal obstruction is measured separately by
\begin{equation}
\mathfrak M^\kappa(T)
=
\left(
\int_\Delta
\sum_{i=1}^{n_p-1}
\left|
\partial_T A_i^\kappa(\s,T)
-
\partial_{s_i}B^\kappa(\s,T)
\right|^2
\,\dd\s
\right)^{1/2}.
\label{eq:DMixed}
\end{equation}
Whereas \(\mathfrak D^\kappa(T)\) quantifies the loss of
saturation-sector integrability on each fixed-temperature slice in continuum
\(\kappa\), \(\mathfrak M^\kappa(T)\) quantifies the mixed
saturation--temperature incompatibility that prevents those slices from
assembling into a single potential on the augmented state domain.

Taken together, the pair
$\big(\mathfrak D^\kappa(T),\mathfrak M^\kappa(T)\big)$
provides a useful continuum-wise decomposition of nonexactness. Small
\(\mathfrak D^\kappa\) and small \(\mathfrak M^\kappa\) indicate behavior
close to full exactness in continuum \(\kappa\). Large \(\mathfrak D^\kappa\)
with small \(\mathfrak M^\kappa\) indicates a regime dominated by
saturation-sector nonintegrability. Small \(\mathfrak D^\kappa\) with large
\(\mathfrak M^\kappa\) indicates a regime that is nearly integrable on each
isothermal slice but remains genuinely nonexact because of thermal coupling.
When both are large, the reduced global-pressure structure is strongly
nonexact in both senses.

These two quantifiers are used below to classify the numerical regimes explored
by the reduced matrix--fracture trajectories and to identify whether the loss
of exactness is dominated by saturation-sector nonintegrability or by mixed
saturation--temperature incompatibility.

%%%%%%%%%%%%%%%%%%%%%%%%%%% PICTURES BENCHMARK A %%%%%%%%%%%%%%%%%%%%%%%%%%%
\begin{figure*}[t]
\centering
\includegraphics[width=0.32\textwidth]{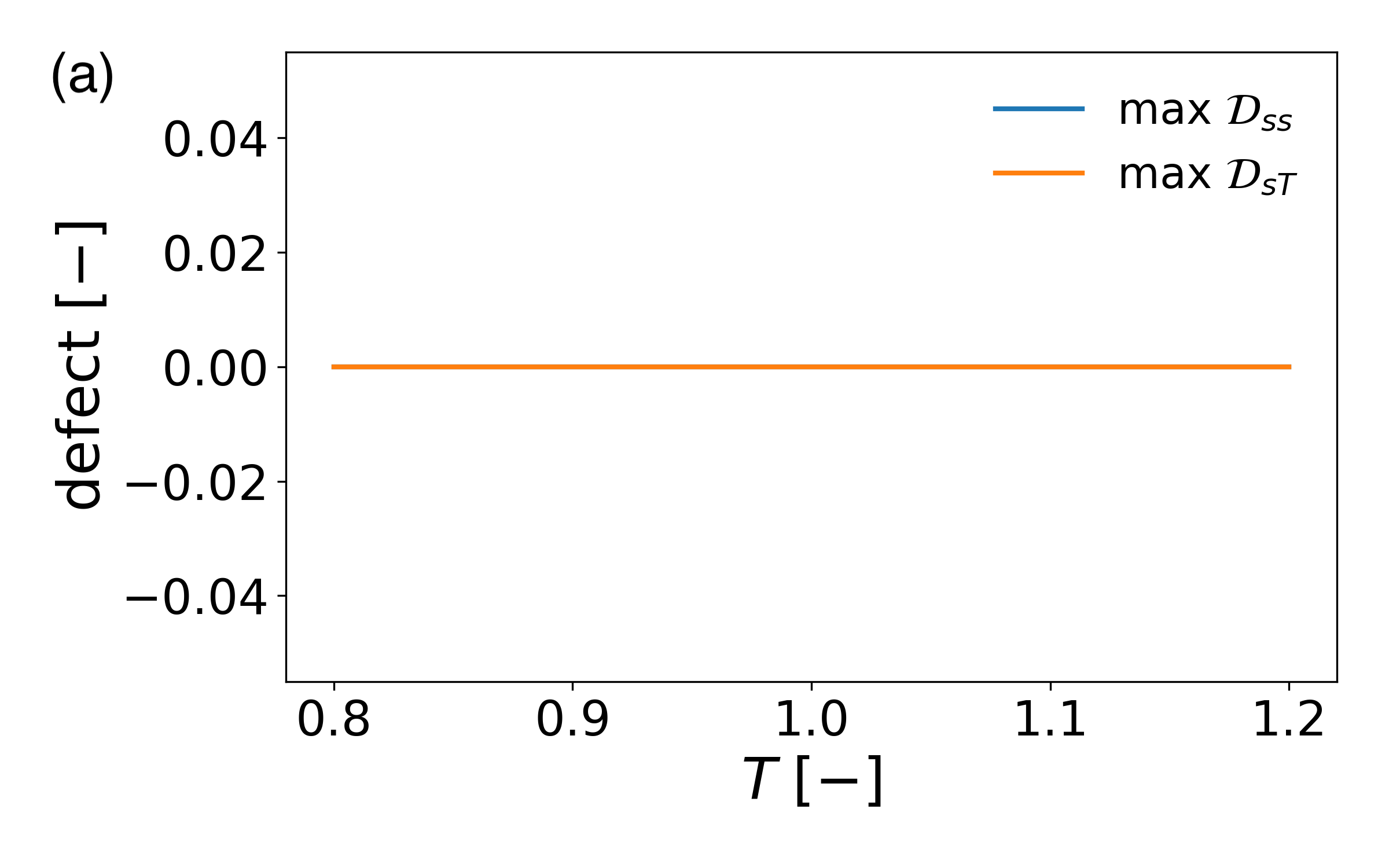}
\includegraphics[width=0.32\textwidth]{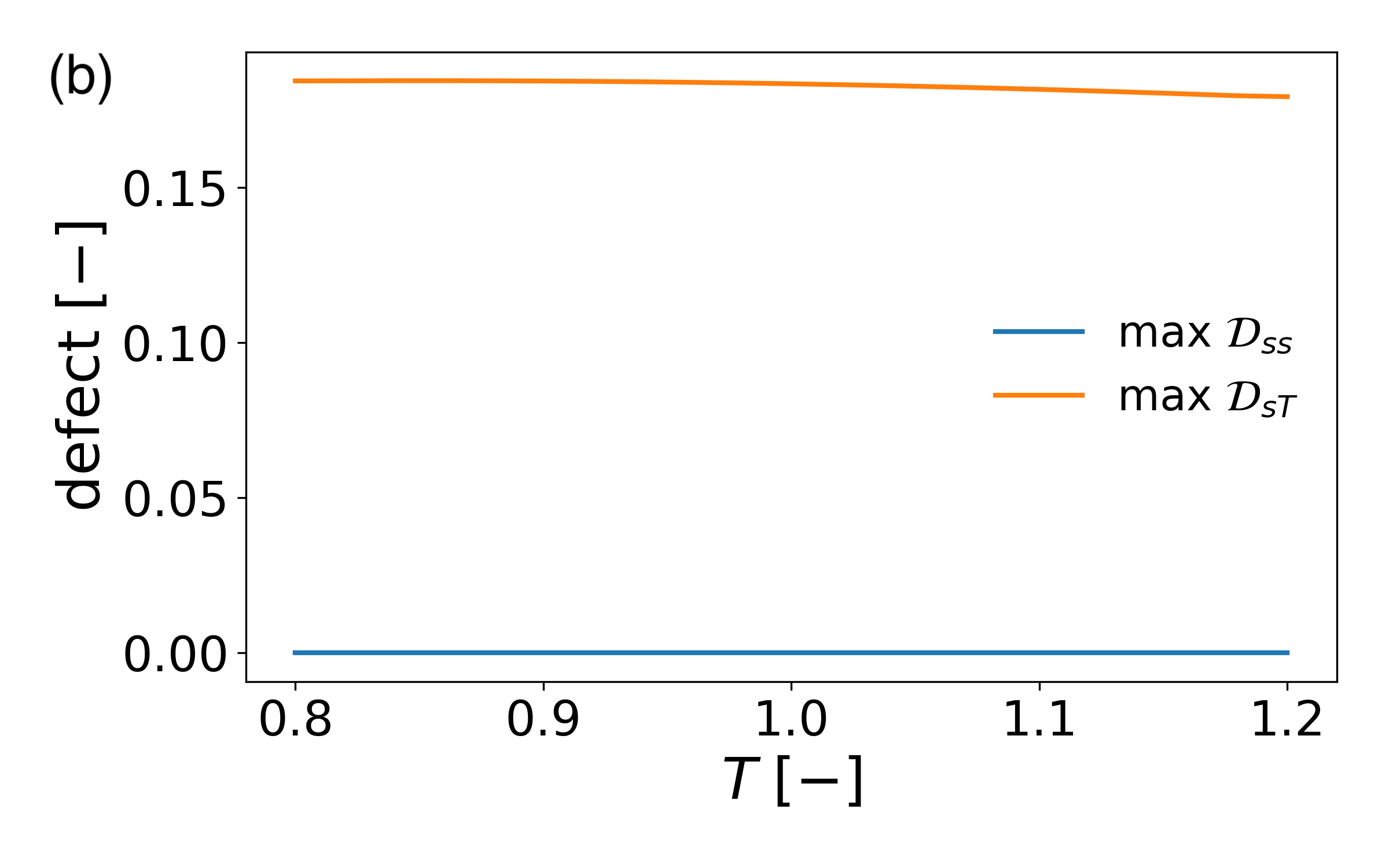}
\includegraphics[width=0.32\textwidth]{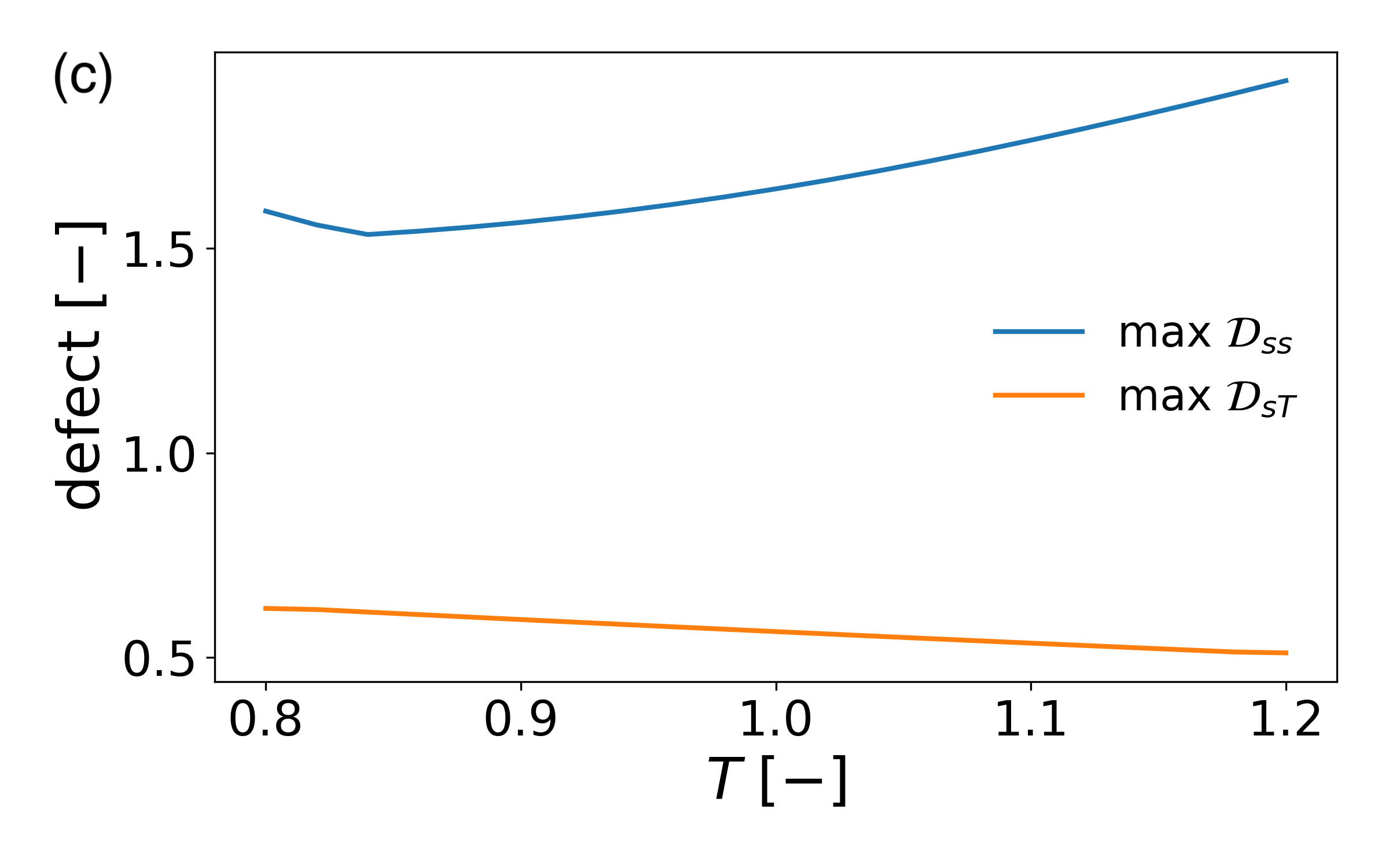}

\vspace{0.3cm}

\includegraphics[width=0.32\textwidth]{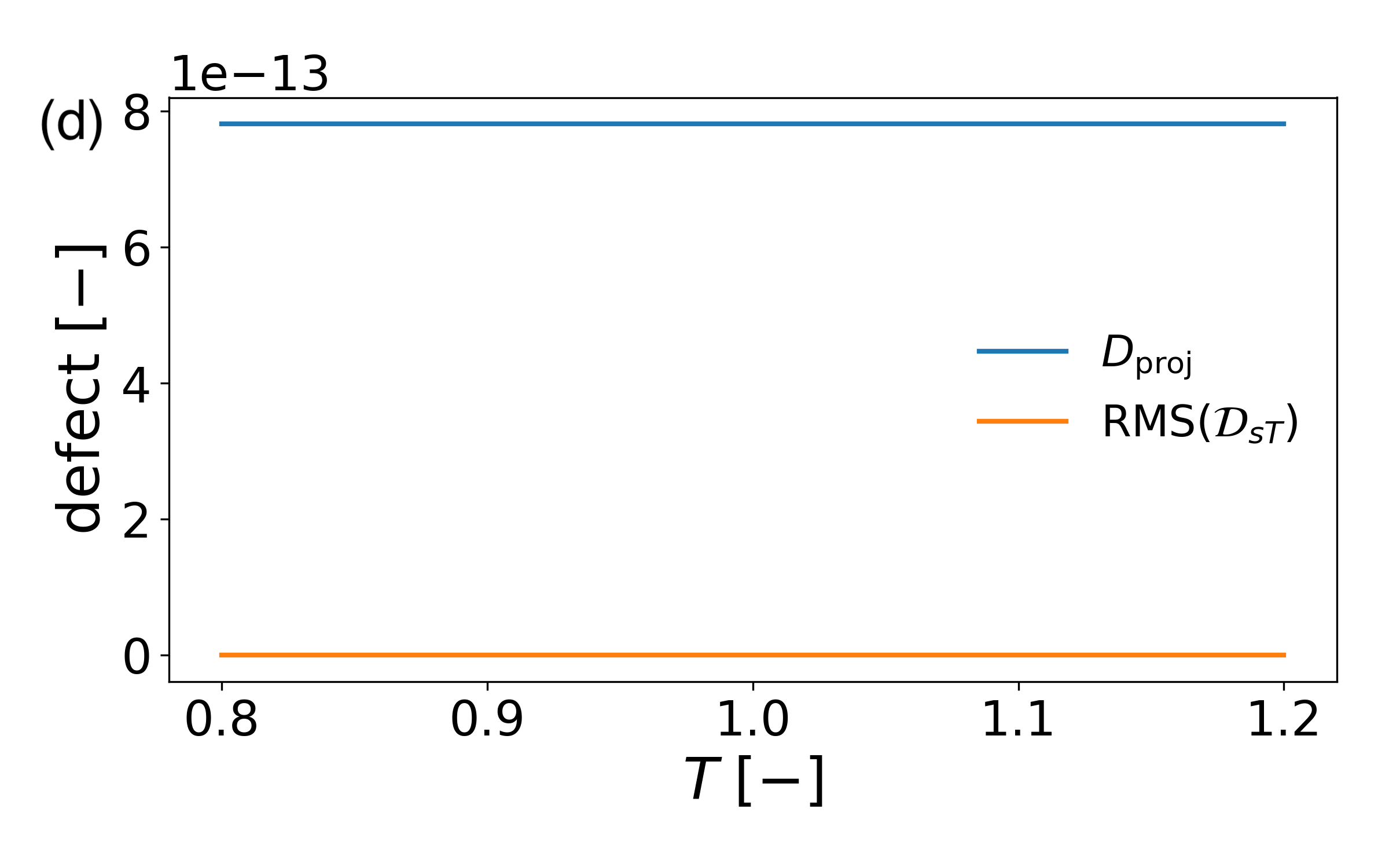}
\includegraphics[width=0.32\textwidth]{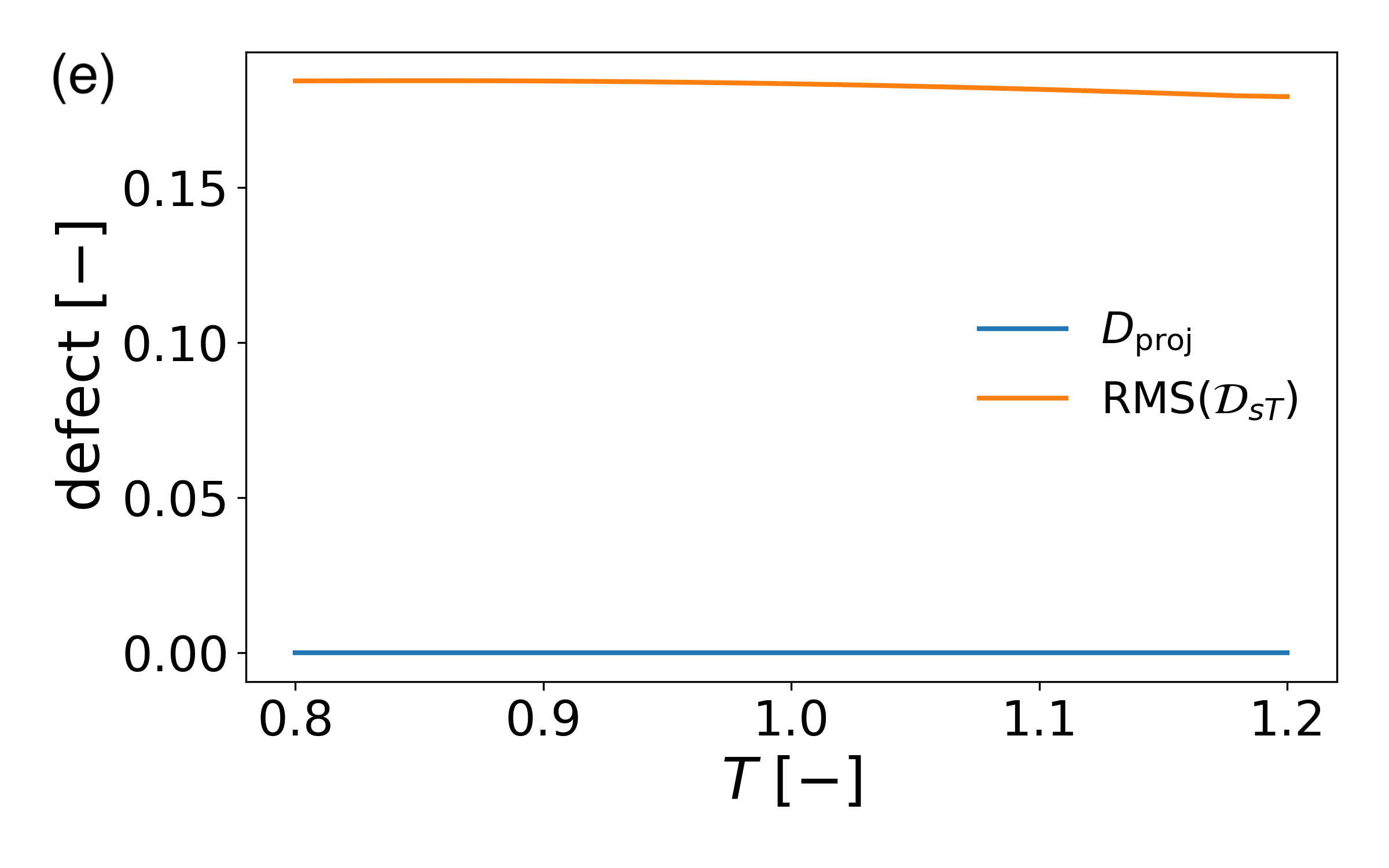}
\includegraphics[width=0.32\textwidth]{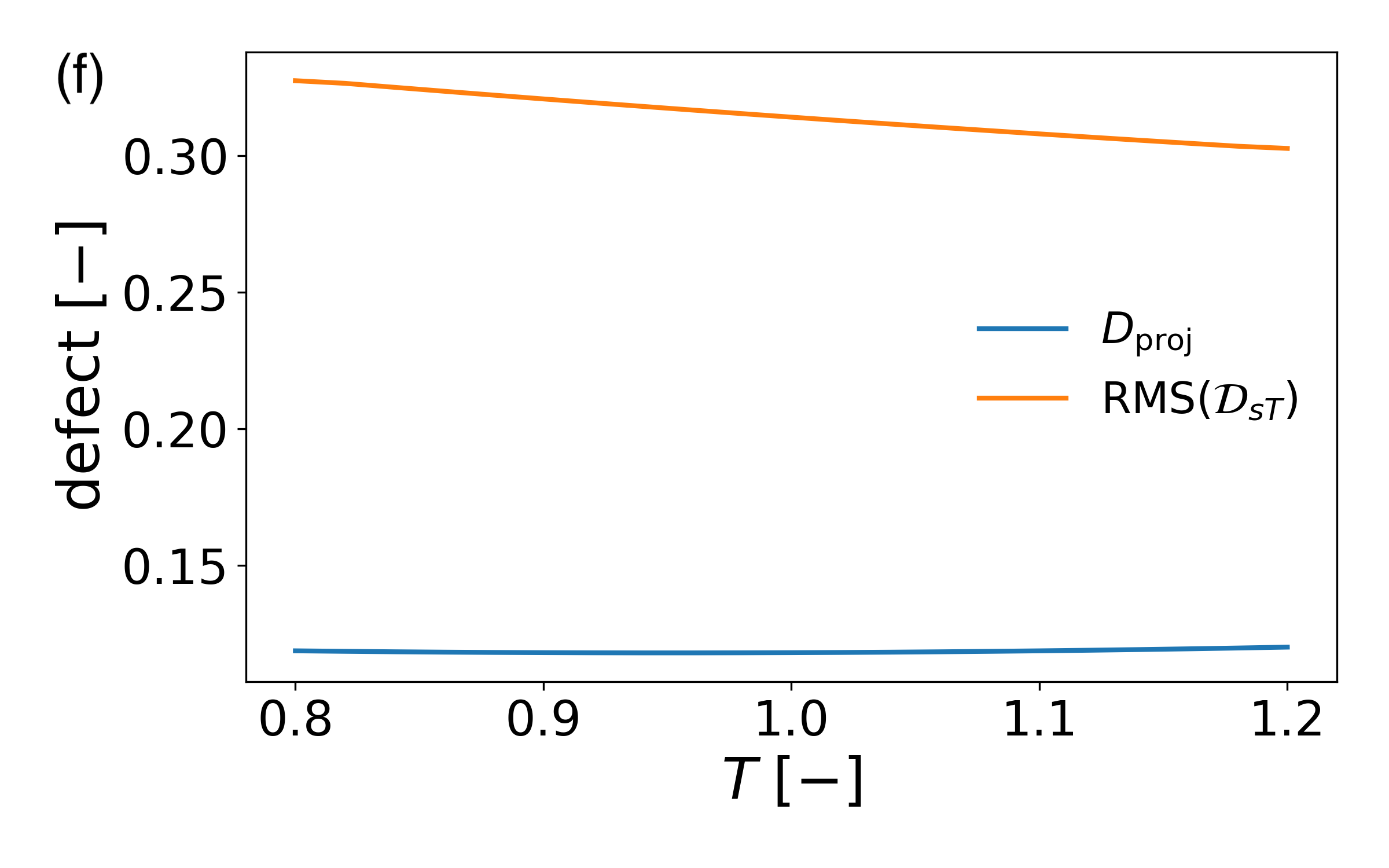}
\caption{
Benchmark~A: line diagnostics on the augmented state domain \(\Delta\times\cT\), with \(T\) as a control parameter. These curves are not time histories. All plotted quantities are absolute dimensionless defect measures. The top row reports the maxima of \(\mathcal D_{ss}\) and \(\mathcal D_{sT}\) over the admissible simplex at each sampled \(T\); the bottom row reports the slice-wise projection defect \(D_{\mathrm{proj}}\) and a discrete RMS-type summary of \(\mathcal D_{sT}\). (a),(d) A1, fully exact: all diagnostics remain at numerical roundoff. (b),(e) A2, slice-wise exact but nonisothermal nonexact: \(\max \mathcal D_{ss}\) and \(D_{\mathrm{proj}}\) remain negligible, whereas \(\max \mathcal D_{sT}\) and its mixed summary remain positive. (c),(f) A3, fully nonexact: both defect sectors are nonzero, and \(D_{\mathrm{proj}}\) is finite. The three columns therefore recover the three regimes predicted by Theorem~\ref{thm:thermal_gtd}.
}
\label{fig:A1A2A3_lines}
\end{figure*}

\begin{figure*}[t]
\centering
\includegraphics[width=0.31\textwidth]{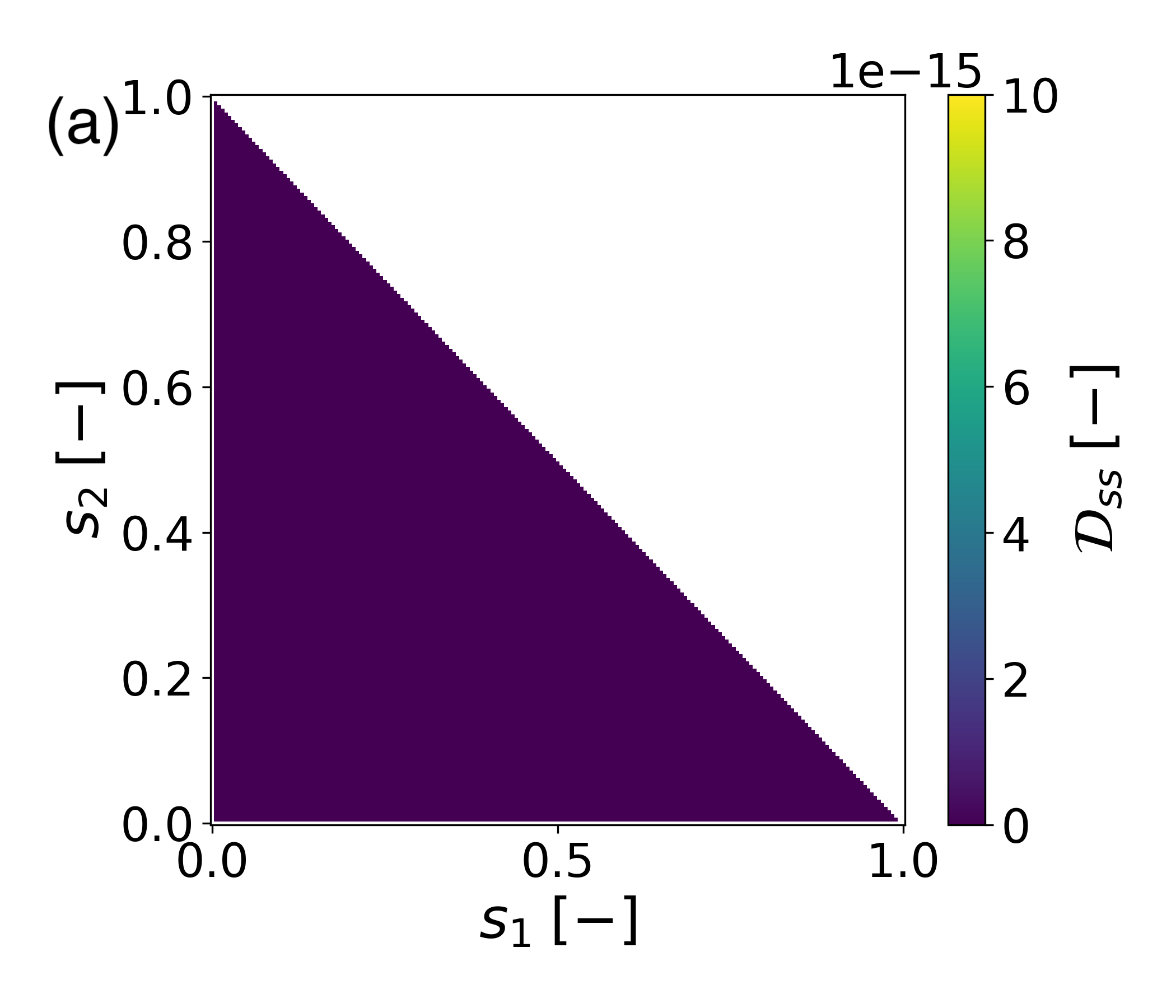}
\includegraphics[width=0.31\textwidth]{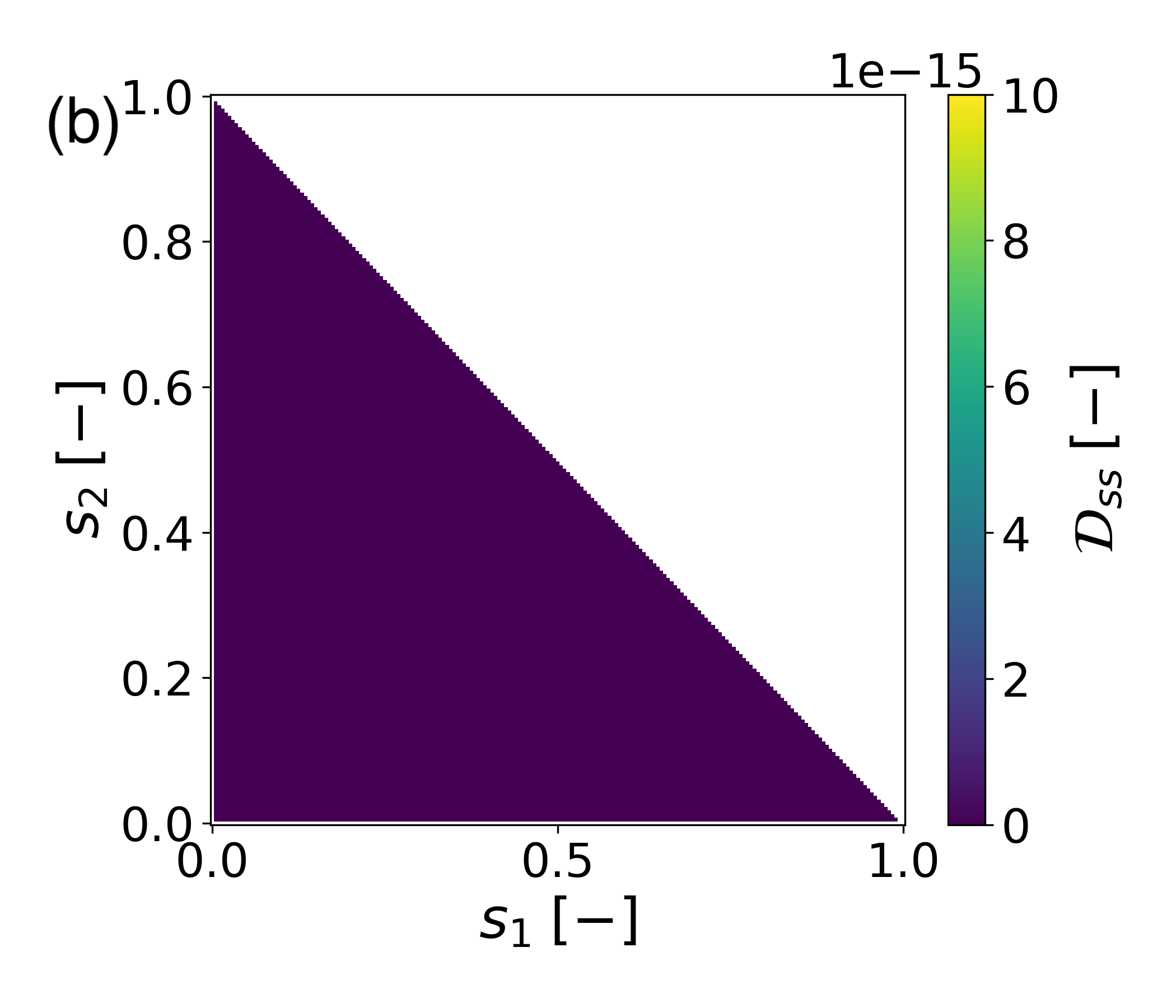}
\includegraphics[width=0.31\textwidth]{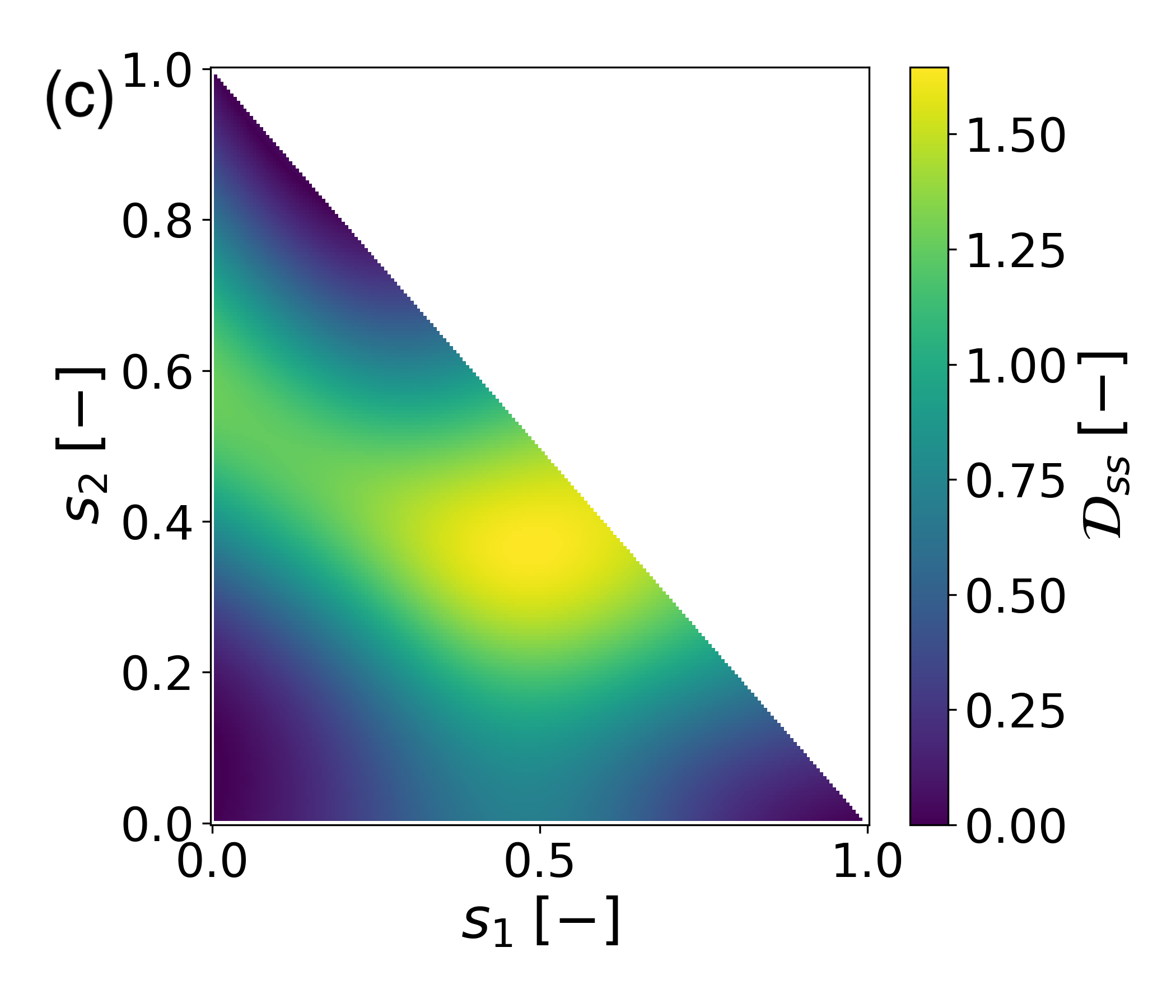}

\vspace{0.3cm}
\includegraphics[width=0.31\textwidth]{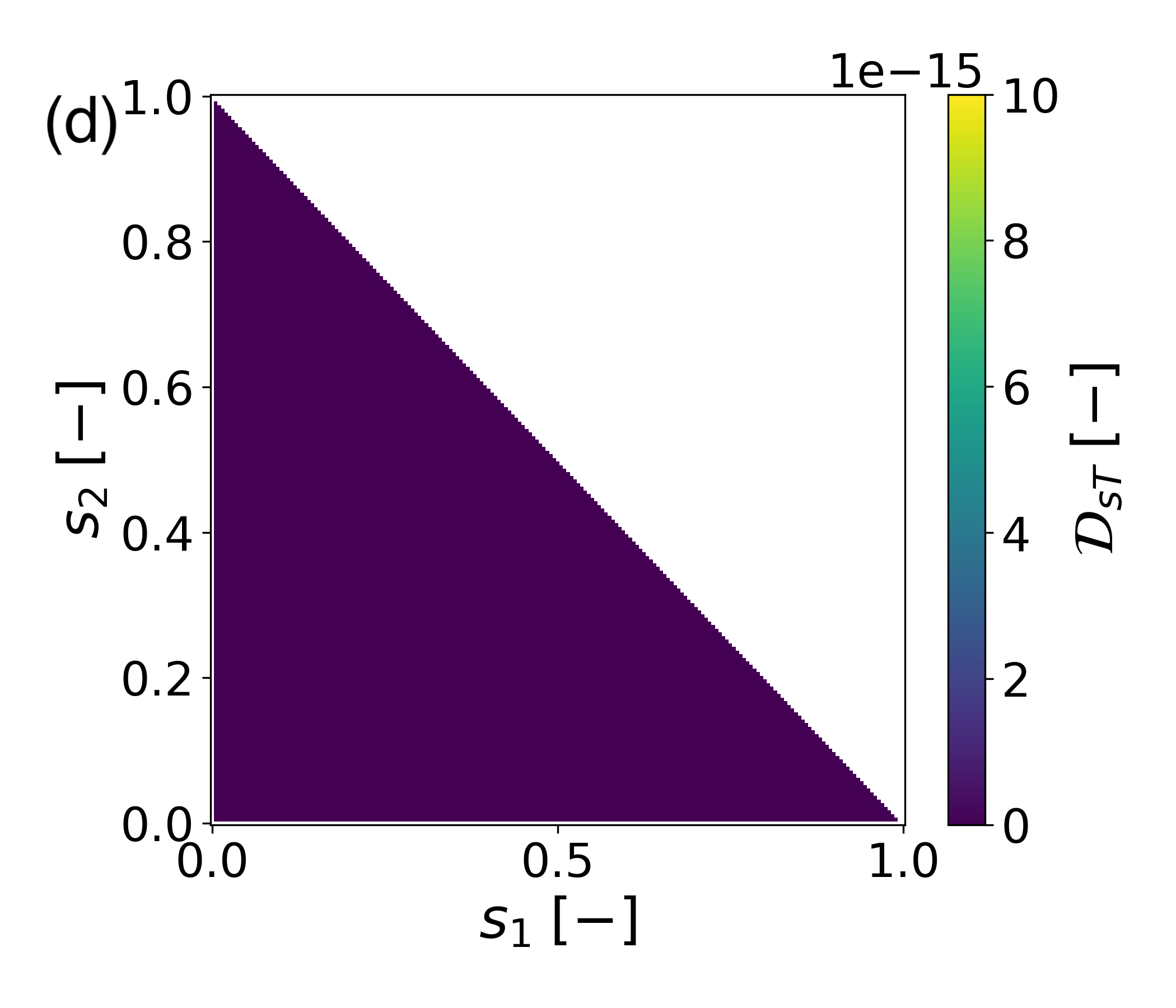}
\includegraphics[width=0.31\textwidth]{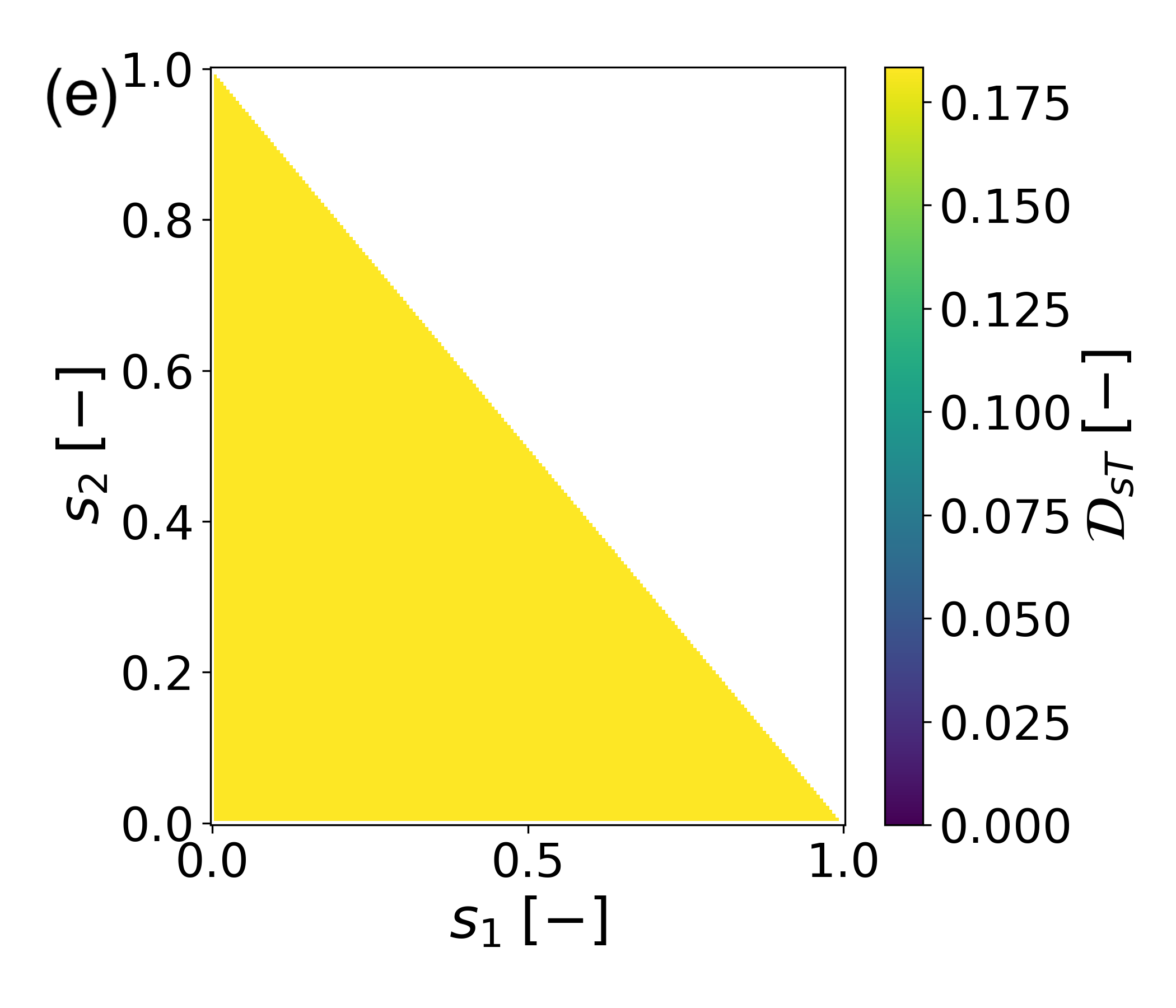}
\includegraphics[width=0.31\textwidth]{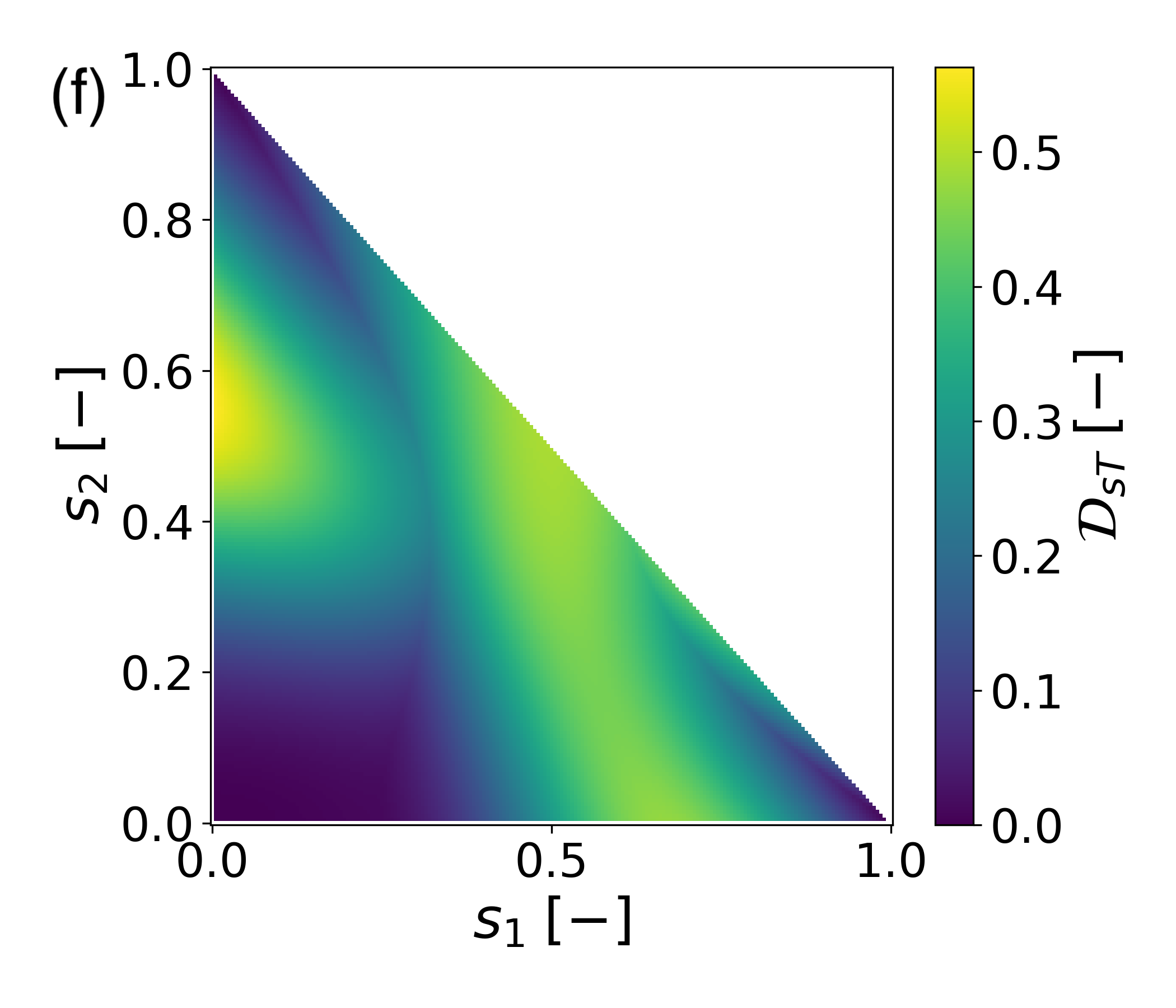}
\caption{
Benchmark~A: representative state-domain defect maps on fixed-temperature slices of \(\Delta\times\cT\). These panels are not time snapshots or stationary PDE solutions, but evaluations of the defect fields on the prescribed constitutive model. (a),(d) A1, fully exact: both \(\mathcal D_{ss}\) and \(\mathcal D_{sT}\) are negligible across the simplex. (b),(e) A2, slice-wise exact but nonisothermal nonexact: \(\mathcal D_{ss}\) remains negligible, whereas \(\mathcal D_{sT}\) is nonzero across the admissible region. (c),(f) A3, fully nonexact: both \(\mathcal D_{ss}\) and \(\mathcal D_{sT}\) are distributed broadly over the simplex. The admissible domain is sampled on a structured Cartesian \((s_1,s_2)\)-grid and masked outside \(\Delta\); the slight staircase appearance of the diagonal boundary is therefore a rendering artifact.
}
\label{fig:A_maps}
\end{figure*}

%%%%%%%%%%%%%%%%%%%%%%%%%%% PICTURES BENCHMARK B %%%%%%%%%%%%%%%%%%%%%%%%%%%

\begin{figure*}[t]
\centering
\includegraphics[width=0.32\textwidth]{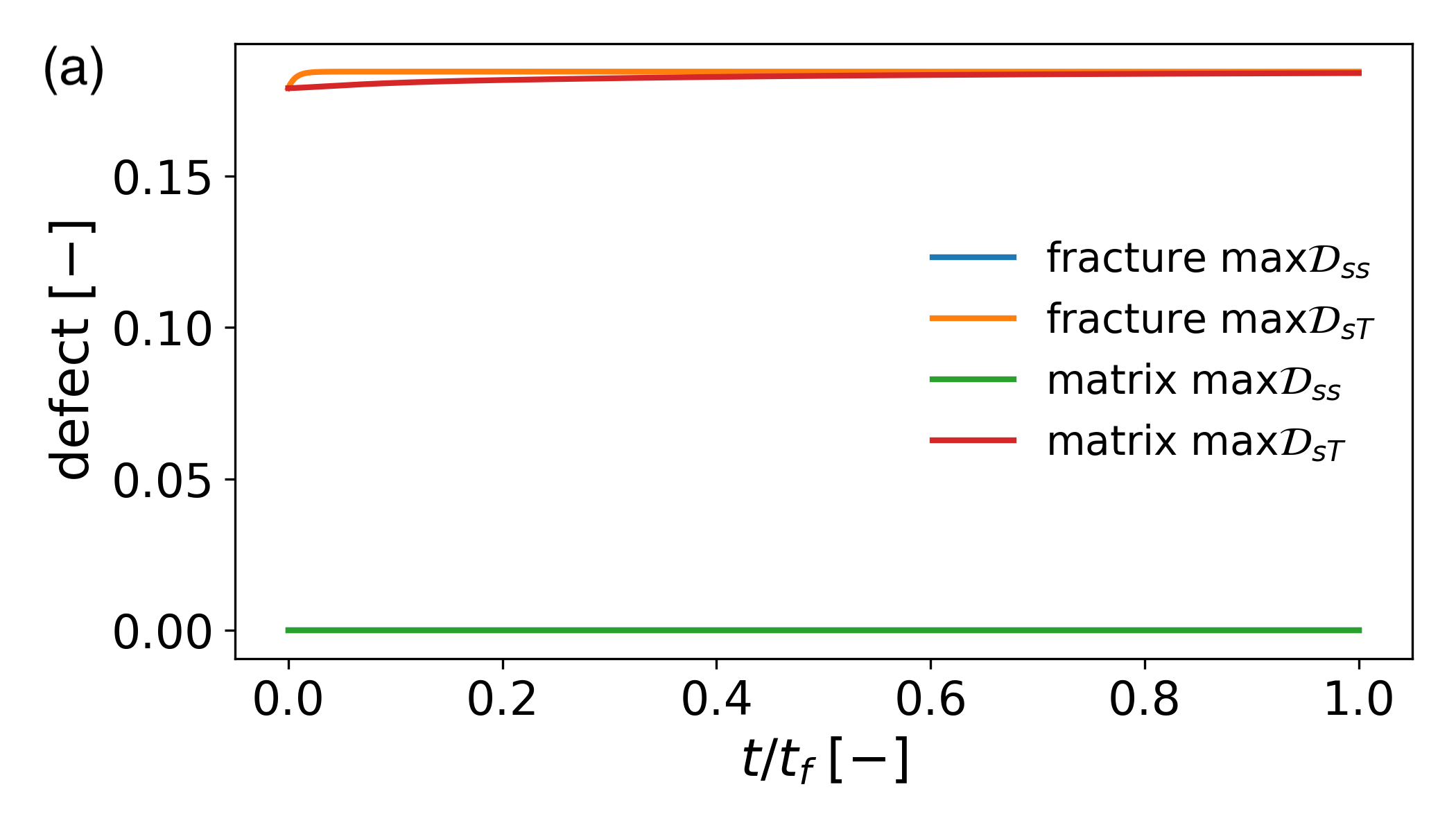}
\includegraphics[width=0.32\textwidth]{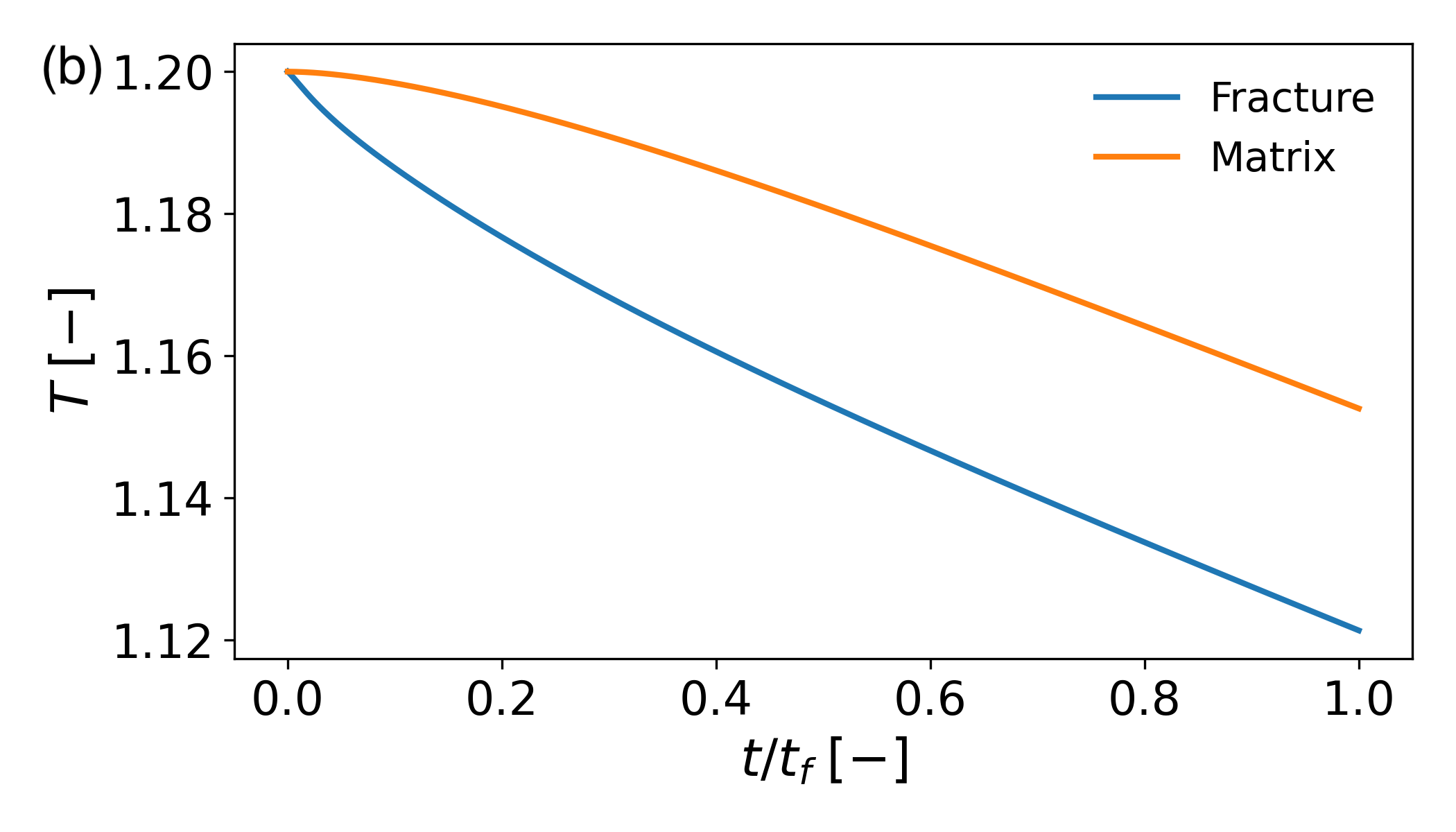}
\includegraphics[width=0.28\textwidth]{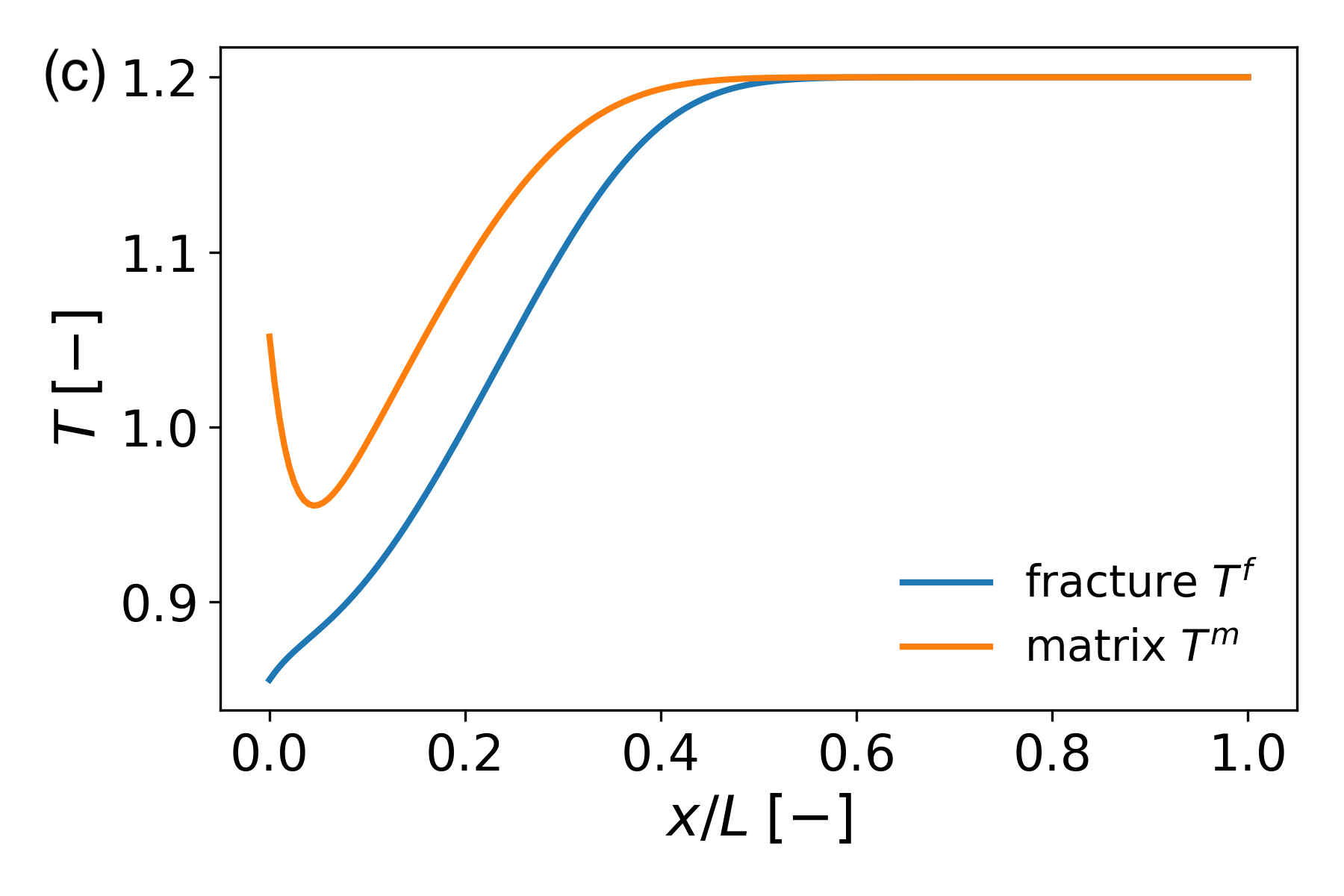}
\vspace{0.3cm}
\includegraphics[width=0.32\textwidth]{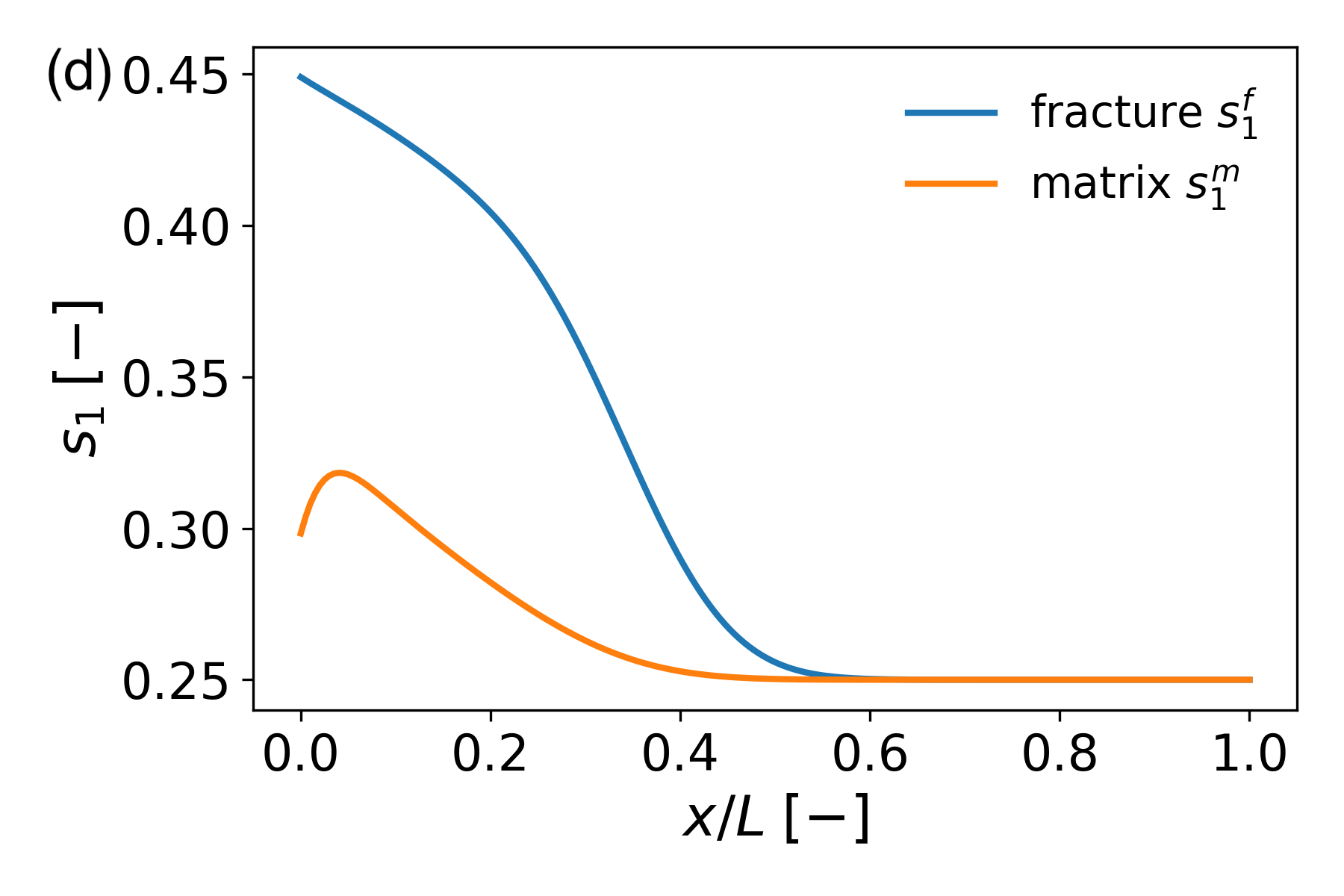}
\includegraphics[width=0.32\textwidth]{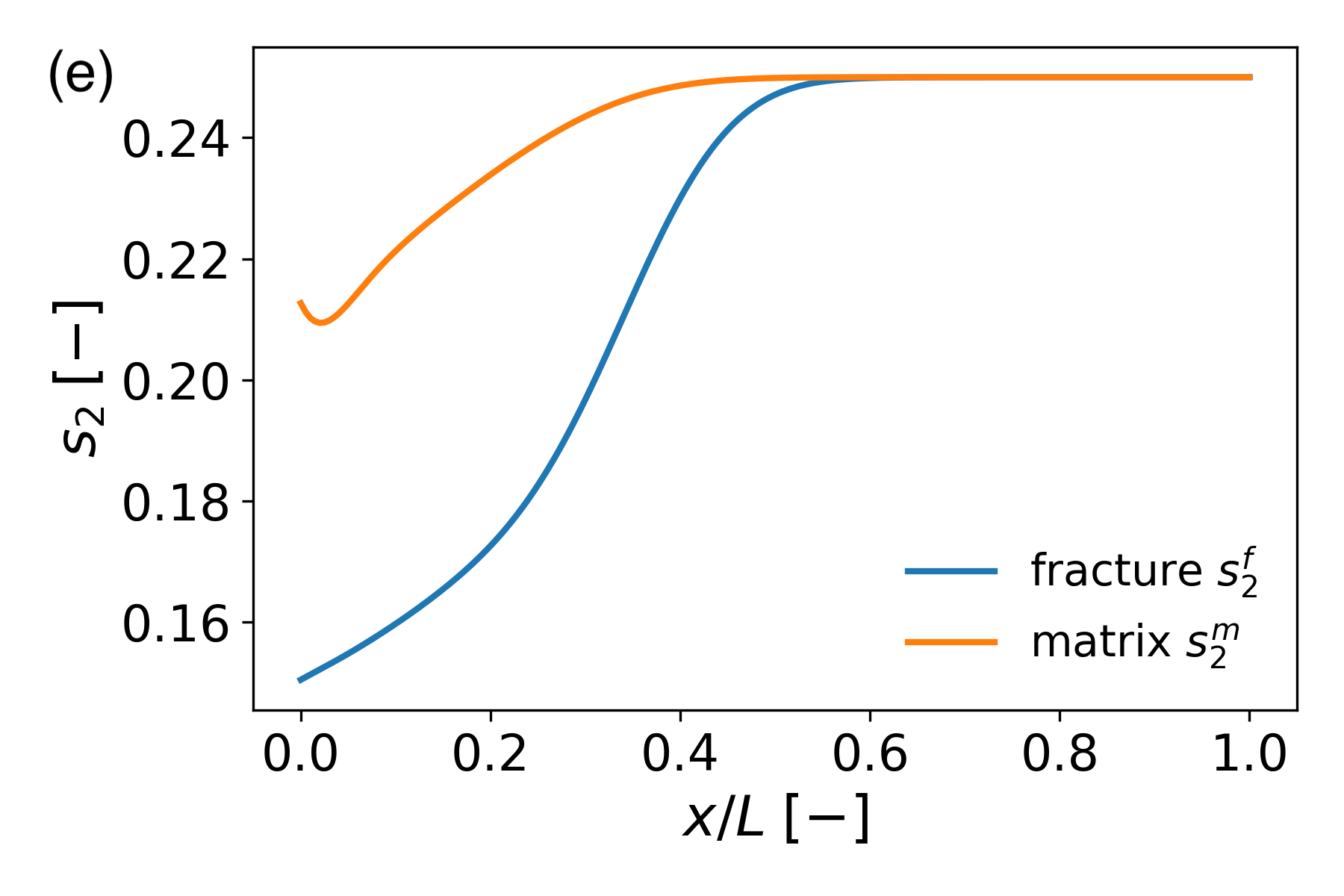}
\includegraphics[width=0.28\textwidth]{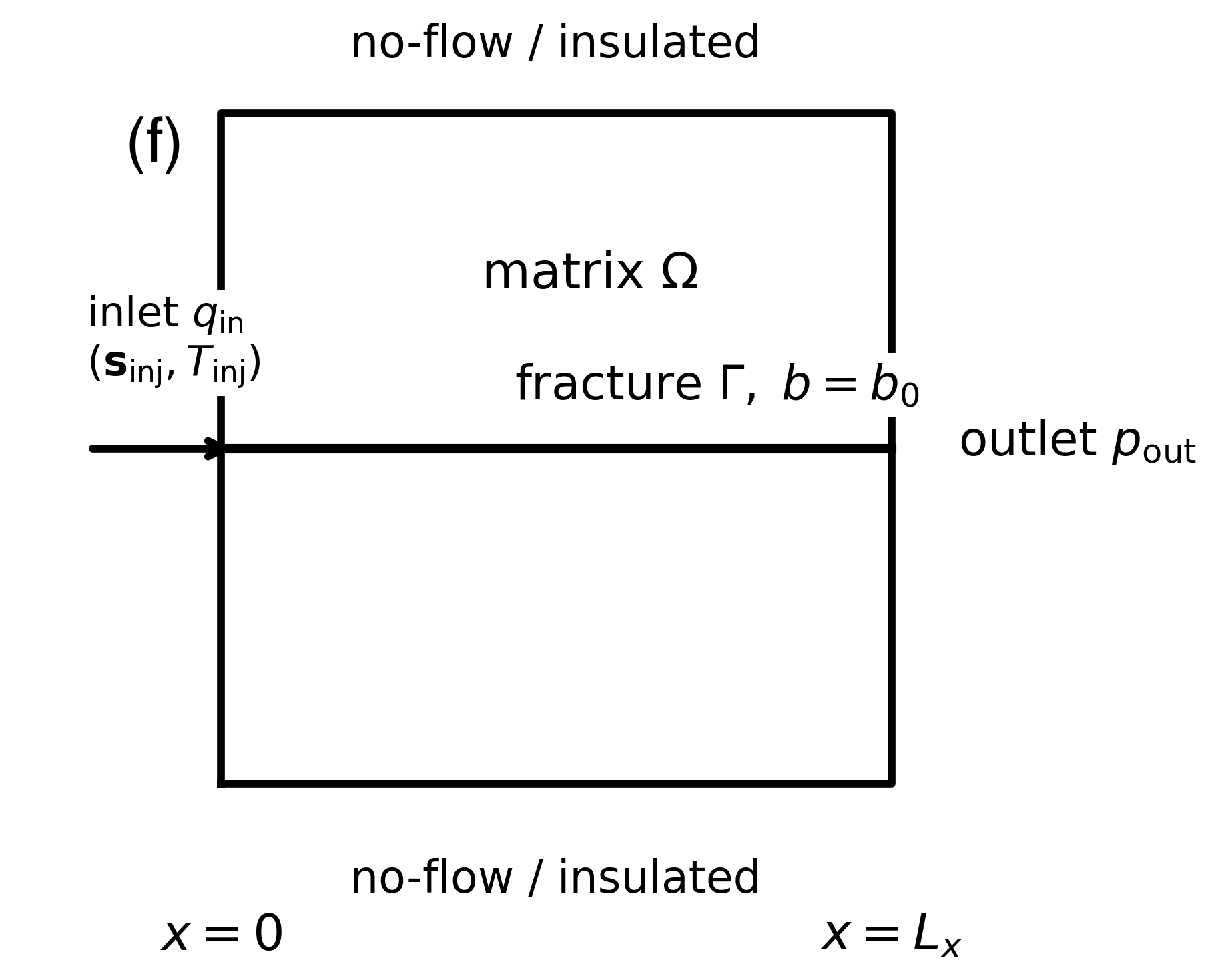}
\caption{
Benchmark~B: fixed-aperture fractured thermal-front benchmark. (a) Time histories of the maximum local defects in matrix and fracture; the mixed defect dominates, while the saturation-sector defect remains smaller. (b) Mean dimensionless temperatures in matrix and fracture versus normalized time. (c) Final temperature profiles \(T(x/L)\). (d) Final \(s_1(x/L)\) profiles. (e) Final \(s_2(x/L)\) profiles. (f) Reduced matrix--fracture geometry and boundary conditions: a rectangular matrix domain \(\Omega\) containing a single embedded fracture \(\Gamma\) with fixed aperture \(b=b_0\), prescribed inflow and injected state \((\s_{\mathrm{inj}},T_{\mathrm{inj}})\) at \(x=0\), outlet pressure at \(x=L_x\), and no-flow, thermally insulated upper and lower boundaries. Even without aperture feedback, the two continua follow different thermodynamic trajectories, and the resulting defect history is dominated by the mixed saturation--temperature sector.
}
\label{fig:benchmarkB_all}
\end{figure*}
%%%%%%%%%%%%%%%%%%%%%%%%%%%%%%%%%%%%%%%%%%%%%%%%%%%%%

\section{Numerical verification and diagnostic validation}
\label{sec:numerics}

This section verifies the augmented-state exactness framework through three
targeted benchmarks.
Benchmark~A verifies the compatibility conditions of
Sec.~\ref{sec:thermal_gtd} directly on prescribed constitutive data.
Benchmark~B validates the behavior of the exactness diagnostics along a reduced
matrix--fracture thermal trajectory with fixed aperture.
Benchmark~C validates the same diagnostic workflow when aperture-sensitive
transmissivity reshapes the state trajectory and the associated defect history.

Throughout this section we restrict attention to \(n_p=3\) phases, so that the saturation state is described by two independent variables \(\s=(s_1,s_2)\) on the simplex \(\Delta\), with \(s_3=1-s_1-s_2\).
The exact global-pressure relation \eqref{eq:pg_exact} and the projected closure \eqref{eq:pg_star} provide the two reference structures used to interpret the diagnostics.

\subsection{Benchmark hierarchy, setup, and conventions}
\label{subsec:verification_philosophy}

The numerical program consists of three benchmarks.

\paragraph*{Benchmark A: prescribed augmented-state data.}
Benchmark~A verifies the compatibility conditions \eqref{eq:compat_ss} and \eqref{eq:compat_sT} directly on \(\Delta\times\cT\), with constitutive data chosen so that exactness or nonexactness is known in advance.
It isolates the state-domain criterion from transport, geometry, and fracture feedback.

\paragraph*{Benchmark B: fractured thermal front with fixed aperture.}
Benchmark~B introduces reduced fractured transport with fixed aperture and tests how matrix--fracture exchange and thermal forcing drive the two continua along different state trajectories.

\paragraph*{Benchmark C: fractured thermal front with prescribed aperture feedback.}
Benchmark~C uses prescribed aperture-feedback histories to test how
aperture-dependent transmissivity modifies the thermal trajectory and the
associated defect history.

Benchmark~A has no physical geometry and is posed directly on
\(\Delta\times\cT\). Benchmarks~B and C are reduced matrix--fracture diagnostics
on \(x\in[0,1]\). They retain separate matrix and fracture states while avoiding
the additional complexity of a production-scale mixed-dimensional simulator.
This reduced setting is sufficient for the present purpose, which is to verify
the augmented-state exactness diagnostics and to show how thermal forcing,
matrix--fracture exchange, and aperture feedback move the system through
different exactness regimes.
The fracture transmissivity used in the reduced benchmarks is the
aperture-dependent quantity $\mathcal T_f(b)$ defined in
Eq.~\eqref{eq:cubic_law}. Throughout this section, $T^f$ denotes fracture
temperature, whereas $\mathcal T_f$ denotes fracture transmissivity.

In Benchmark~B the aperture is fixed, whereas Benchmark~C uses prescribed
aperture-feedback histories to isolate the effect of aperture evolution on
transmissivity, velocity, temperature profiles, and defect histories.

All temperatures shown in the plots are nondimensional.
In Benchmark~A, \(T\) is a dimensionless control parameter on \(\Delta\times\cT\).
In Benchmarks~B--C, \(T\) is the corresponding reduced nondimensional temperature used in the surrogate models.

\subsection{Implementation strategy and diagnostics}
\label{subsec:discrete_strategy}

The three benchmarks are implemented at different levels.
Benchmark~A is a direct state-domain evaluation on sampled points of \(\Delta\times\cT\), with no physical PDE.
Benchmark~B is a reduced one-dimensional fractured-transport surrogate combining upwind advection, explicit smoothing, and matrix--fracture exchange.
Benchmark~C uses prescribed aperture and thermal histories to isolate the
feedback between aperture, transmissivity, velocity, and exactness defects.

At each evaluation point, the implementation computes the constitutive quantities introduced in Sec.~\ref{sec:thermal_gtd}, namely the mobilities, capillary pressures, and the fields \(A_i\) and \(B\), and then assembles the local defect fields \eqref{eq:defect_ss} and \eqref{eq:defect_sT}.
In Benchmark~A the required derivatives come directly from the prescribed constitutive expressions.
In Benchmarks~B--C they are evaluated numerically along the sampled state evolution using consistent finite-difference approximations.

Only the diagnostics used in the reported results are retained.

\paragraph*{(i) Local closure defects.}
The primary diagnostics are the local defect fields \eqref{eq:defect_ss} and \eqref{eq:defect_sT}.
Their maxima over space and time are recorded separately in matrix and fracture when applicable.
These measure the loss of saturation-sector compatibility and mixed saturation--temperature compatibility, respectively.

\paragraph*{(ii) Path-integral defect.}
A second diagnostic is based on the augmented capillary field
\(\mathbf C=(A_1,\ldots,A_{n_p-1},B)\).
For two piecewise \(C^{1}\) paths
\(\gamma_\ell(\tau)=\big(\s_\ell(\tau),T_\ell(\tau)\big)\), where
\(\ell=1,2\), with common endpoints
\(\gamma_1(0)=\gamma_2(0)\) and
\(\gamma_1(1)=\gamma_2(1)\), let \(\mathbf y=(\s,T)\). We define the path
defect by
\begin{equation}
\delta_C(\gamma_1,\gamma_2)
:=
\left|
\int_{\gamma_1}\mathbf C\cdot \mathrm d\mathbf y
-
\int_{\gamma_2}\mathbf C\cdot \mathrm d\mathbf y
\right|.
\label{eq:path_defect}
\end{equation}
Thus \(\delta_C(\gamma_1,\gamma_2)\) measures the mismatch between two
state-domain integrals connecting the same endpoints. It vanishes up to
numerical tolerance if \(\mathbf C\) is conservative and is used primarily in
Benchmark~A, where the comparison paths can be prescribed directly.

\paragraph*{(iii) Slice-wise projection and mixed summaries.}
In Benchmark~A we also monitor the slice-wise projection diagnostic, together
with a discrete RMS-type slice-wise summary of the local mixed defect
\(\mathcal D_{sT}\). These separate saturation-sector nonintegrability from
genuinely nonisothermal loss of exactness.

\paragraph*{(iv) Thermal, saturation, and aperture histories.}
In Benchmarks~B and C, mean temperatures, final temperature profiles, selected saturation profiles, and, in Benchmark~C, aperture, transmissivity, and fracture-velocity histories are reported to make the evolving state trajectory visible and to connect it to the defect histories.

The role of these benchmarks is therefore specific. They validate the
augmented-state exactness diagnostics and their behavior along reduced
nonisothermal matrix–fracture trajectories. They are not presented as a full
verification and validation campaign for a production multiphase-flow solver.
A full solver study would require an independent phase-pressure implementation,
pressure and phase-flux comparisons, mass and energy balance errors, and
grid/time-step convergence tests. Those tasks are distinct from the diagnostic
validation performed here and are identified below as future work.

To make the manufactured and reduced benchmarks reproducible, we specify the
constitutive and numerical data used below. Benchmark A is evaluated on a
uniform simplex grid with \(201\times201\) sampling points before masking the
inadmissible region \(s_1+s_2>1\). The temperature interval is sampled by 21
uniform values in \(T\in[0.8,1.2]\). 
We write $s_3=1-s_1-s_2$.

The manufactured constitutive families are generated from the common form
\begin{align}
\lambda_\alpha(\s,T)
&=
\ell_\alpha q_\alpha(\s)
\exp\!\left[\beta_\alpha(T-T_0)\right],
\label{eq:benchA_lam_common}
\\
p_{c,1}(\s,T)
&=
a_1s_1+a_2s_2+a_3s_1s_2 \nonumber \\
&+
(T-T_0)\left(c_1s_1^2+c_2s_1s_2\right),
\label{eq:benchA_pc1_common}
\\
p_{c,2}(\s,T)
&=
b_1s_1+b_2s_2+b_3s_1s_2 \nonumber \\
&+
(T-T_0)\left(d_1s_2^2+d_2s_1s_2\right),
\label{eq:benchA_pc2_common}
\\
p_{c,3}(\s,T)&=0 .
\label{eq:benchA_pc3_common}
\end{align}
Here \(T_0=1\). For A1 and A2, \(q_1=q_2=q_3=1\). For A3,
\(q_1=s_1^2\), \(q_2=s_2^2\), and \(q_3=s_3^2\). The parameter values used in
Eqs.~\eqref{eq:benchA_lam_common}--\eqref{eq:benchA_pc3_common} are reported in
Table~\ref{tab:benchmarkA_parameters}. The finite-difference diagnostics are
computed only where the required neighboring points remain inside the simplex,
so the reported fields exclude degenerate boundary stencils.

\begin{table}[htbp]
\caption{Manufactured constitutive families used in Benchmark A.}
\label{tab:benchmarkA_parameters}
\centering
\scriptsize
\begin{tabular}{ll}
\hline
Case & Data \\
\hline
A1 &
\begin{tabular}[t]{@{}l@{}}
Class: exact \\
\((\ell_1,\ell_2,\ell_3)=(1,1.2,0.8)\) \\
\((\beta_1,\beta_2,\beta_3)=(0,0,0)\) \\
\(q_1=q_2=q_3=1\) \\
\((a_1,a_2,a_3)=(1,0.4,0)\) \\
\((b_1,b_2,b_3)=(-0.2,0.8,0)\) \\
\((c_1,c_2)=(0,0)\), \((d_1,d_2)=(0,0)\)
\end{tabular}
\\
\hline
A2 &
\begin{tabular}[t]{@{}l@{}}
Class: mixed nonexact \\
\((\ell_1,\ell_2,\ell_3)=(1,1.2,0.8)\) \\
\((\beta_1,\beta_2,\beta_3)=(0,0.8,-0.6)\) \\
\(q_1=q_2=q_3=1\) \\
\((a_1,a_2,a_3)=(1,0.4,0)\) \\
\((b_1,b_2,b_3)=(-0.2,0.8,0)\) \\
\((c_1,c_2)=(0,0)\), \((d_1,d_2)=(0,0)\)
\end{tabular}
\\
\hline
A3 &
\begin{tabular}[t]{@{}l@{}}
Class: fully nonexact \\
\((\ell_1,\ell_2,\ell_3)=(1,1.3,1.1)\) \\
\((\beta_1,\beta_2,\beta_3)=(0,0.5,-0.4)\) \\
\((q_1,q_2,q_3)=(s_1^2,s_2^2,s_3^2)\) \\
\((a_1,a_2,a_3)=(1,0.5,0.7)\) \\
\((b_1,b_2,b_3)=(0.2,1.1,-0.5)\) \\
\((c_1,c_2)=(0.8,0.3)\), \((d_1,d_2)=(-0.6,0.4)\)
\end{tabular}
\\
\hline
\end{tabular}
\end{table}

Benchmark B uses the mixed-nonexact constitutive family A2 and evolves one
matrix profile and one fracture profile on \(x\in[0,1]\). The numerical
parameters are listed in Table~\ref{tab:benchmarkB_summary}. The local
exactness defects are evaluated with finite differences using
\(\epsilon_s=10^{-6}\) and \(\epsilon_T=10^{-4}\), with simplex clipping at
\(10^{-8}\).

Benchmark C uses prescribed synthetic aperture-feedback histories to isolate
the effect of aperture evolution on transmissivity, mean velocity, temperature
profiles, and defect histories. The coordinates are nondimensional:
\(\hat x=x/L\) and \(\hat t=t/t_f\). The parameters defining the prescribed
profiles are summarized in Table~\ref{tab:benchmarkC_summary}.

The reduced benchmarks can be interpreted through nondimensional ratios that
control the sampled flow regime. In Benchmark B, the velocity contrast
\(v_f/v_m=17.5\) gives a fracture-dominated advective thermal front, while
\(\eta_h=2.0\) and \(\eta_s=0.5\) control thermal and saturation exchange
between the two continua. In Benchmark C, aperture feedback enters through the
cubic law \(\mathcal T_f=b^3/12\), so an aperture ratio \(b/b_0\) produces a
transmissivity ratio \((b/b_0)^3\). These quantities determine how strongly
thermal advection, matrix--fracture exchange, and aperture-sensitive
transmissivity move the local state through the augmented
saturation--temperature domain.

\subsection{Benchmark A: prescribed augmented-state exactness tests}
\label{subsec:benchmark_A_num}

Benchmark~A is not a transport simulation.
It is a direct verification on \(\Delta\times\cT\), with \(T\) acting as a control parameter.
The line plots and maps are therefore evaluations of the prescribed constitutive model along sampled temperature paths and on fixed-temperature slices.
For visualization, the state-domain maps are evaluated on a structured Cartesian sampling of the \((s_1,s_2)\)-plane and masked outside the admissible simplex \(\Delta\).

Figure~\ref{fig:A1A2A3_lines} reports the line diagnostics for the three prescribed regimes, and Fig.~\ref{fig:A_maps} reports the corresponding state-domain defect maps at representative sampled values of \(T\).

The corresponding numerical defect values are summarized in
Table~\ref{tab:benchmarkA_diagnostics}.

\begin{table}[htbp]
\caption{Summary of Benchmark~A exactness diagnostics computed over
($T\in[0.8,1.2]$) on the masked simplex grid. All reported quantities are
dimensionless diagnostic measures.}
\label{tab:benchmarkA_diagnostics}
\centering
\scriptsize
\begin{tabular}{llll}
\hline
Case & Path defect & \(\max\mathcal D_{ss}\) & \(\max\mathcal D_{sT}\) \\
\hline
A1 & \(4.16\times10^{-17}\) & \(0\) & \(0\) \\
A2 & \(5.06\times10^{-3}\) & \(0\) & \(1.84\times10^{-1}\) \\
A3 & \(1.13\times10^{-2}\) & \(1.91\) & \(6.20\times10^{-1}\) \\
\hline
\end{tabular}
\end{table}

The magnitude of the diagnostics provides the regime classification used below.
Roundoff-level values, as in A1, indicate exactness within numerical tolerance.
A vanishing \(\mathcal D_{ss}\) combined with finite \(\mathcal D_{sT}\), as in
A2, identifies a regime that is exact on isothermal slices but nonexact on the
augmented state domain. Simultaneously finite \(\mathcal D_{ss}\) and
\(\mathcal D_{sT}\), as in A3, indicates full loss of exactness. In Benchmarks
B and C, the same diagnostic hierarchy is evaluated along the reduced
matrix--fracture trajectories.

\paragraph*{Case A1: fully exact data.}
In Figs.~\ref{fig:A1A2A3_lines}(a) and \ref{fig:A1A2A3_lines}(d), all reported diagnostics remain at numerical roundoff level across the sampled temperature interval.
Figs.~\ref{fig:A_maps}(a) and \ref{fig:A_maps}(d) show the same result over the simplex: no localized nonexact region appears anywhere in the sampled state domain.

\paragraph*{Case A2: slice-wise exact but nonisothermal nonexact.}
Case A2 is the benchmark that isolates the difference between the usual
isothermal treatment and the present nonisothermal framework. An isothermal
TD/gTD test performed at fixed \(T\) would classify this case as exact because
\(\mathcal D_{ss}=0\). The nonisothermal augmented-state test detects the
finite mixed defect \(\mathcal D_{sT}\), showing that fixed-temperature
exactness is not sufficient for full nonisothermal exactness.
In Figs.~\ref{fig:A1A2A3_lines}(b) and \ref{fig:A1A2A3_lines}(e), the saturation-sector defect and the slice-wise projection diagnostic remain negligible, whereas the mixed defect and its discrete RMS-type summary remain clearly nonzero.
Figs.~\ref{fig:A_maps}(b) and \ref{fig:A_maps}(e) show the same separation geometrically: \(\mathcal D_{ss}\) remains negligible, while \(\mathcal D_{sT}\) is finite across the admissible region.
Each fixed-temperature slice is therefore exact, but exactness is lost on the full augmented state domain once temperature variation is included.

\paragraph*{Case A3: fully nonexact data.}
In Figs.~\ref{fig:A1A2A3_lines}(c) and \ref{fig:A1A2A3_lines}(f), both defect sectors are clearly nonzero, and the slice-wise projection diagnostic is likewise finite.
Figs.~\ref{fig:A_maps}(c) and \ref{fig:A_maps}(f) show broad, structured nonexactness over the simplex.
Unlike A2, the loss of exactness is no longer purely mixed: it is already present within fixed-temperature slices, which is why the projection defect no longer vanishes.

These three cases reproduce the hierarchy predicted by the theory: A1 is fully exact, A2 is slice-wise exact but mixed nonexact, and A3 is fully nonexact.
Benchmark~A therefore validates both the exactness theorem and the defect decomposition used in the fractured tests.

%%%%%%%%%%%%%%%%%%%%%%%%%%%%%%%%%%%%%%%%
\begin{figure*}[t]
\centering
\includegraphics[width=0.32\textwidth]{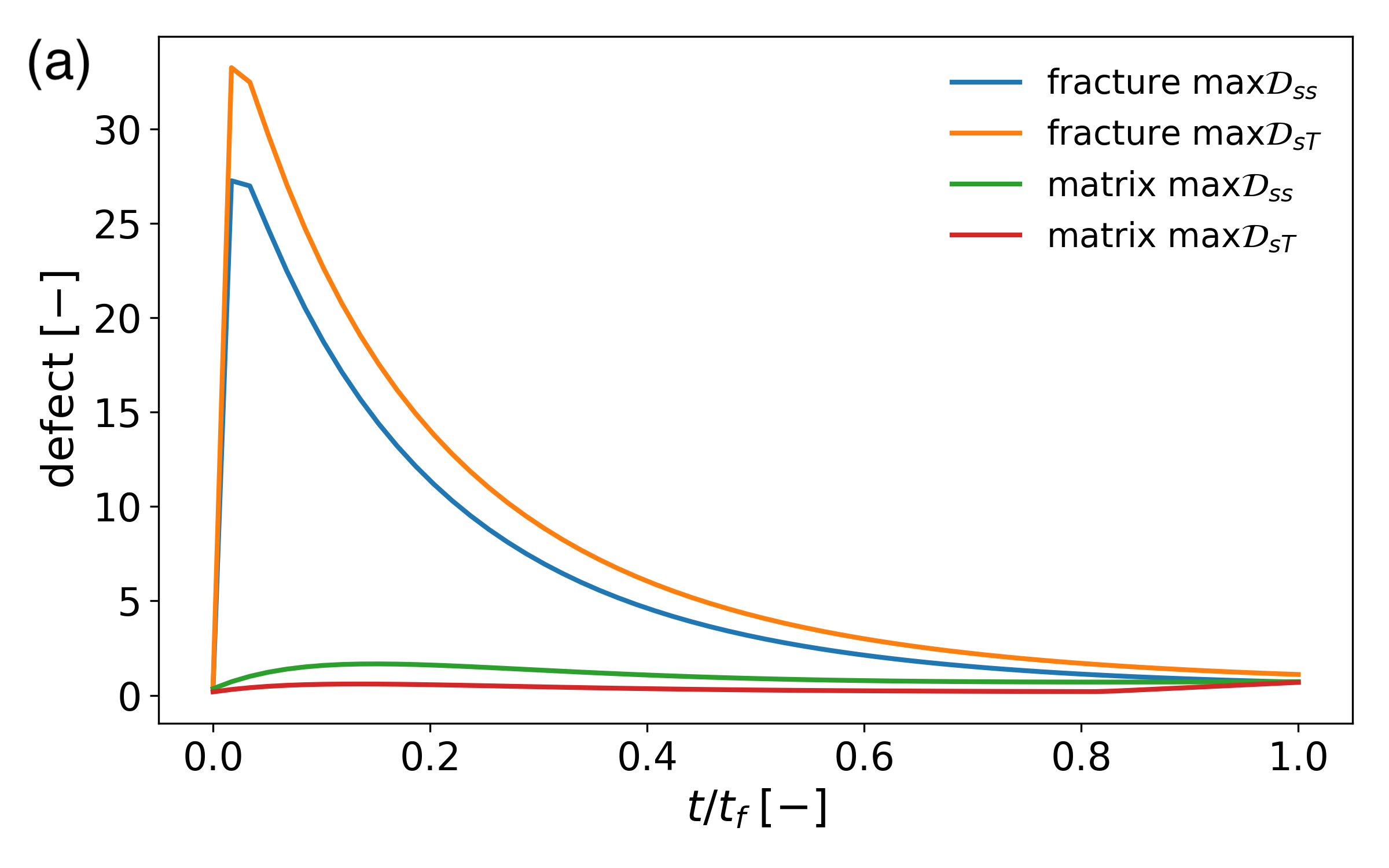}
\includegraphics[width=0.32\textwidth]{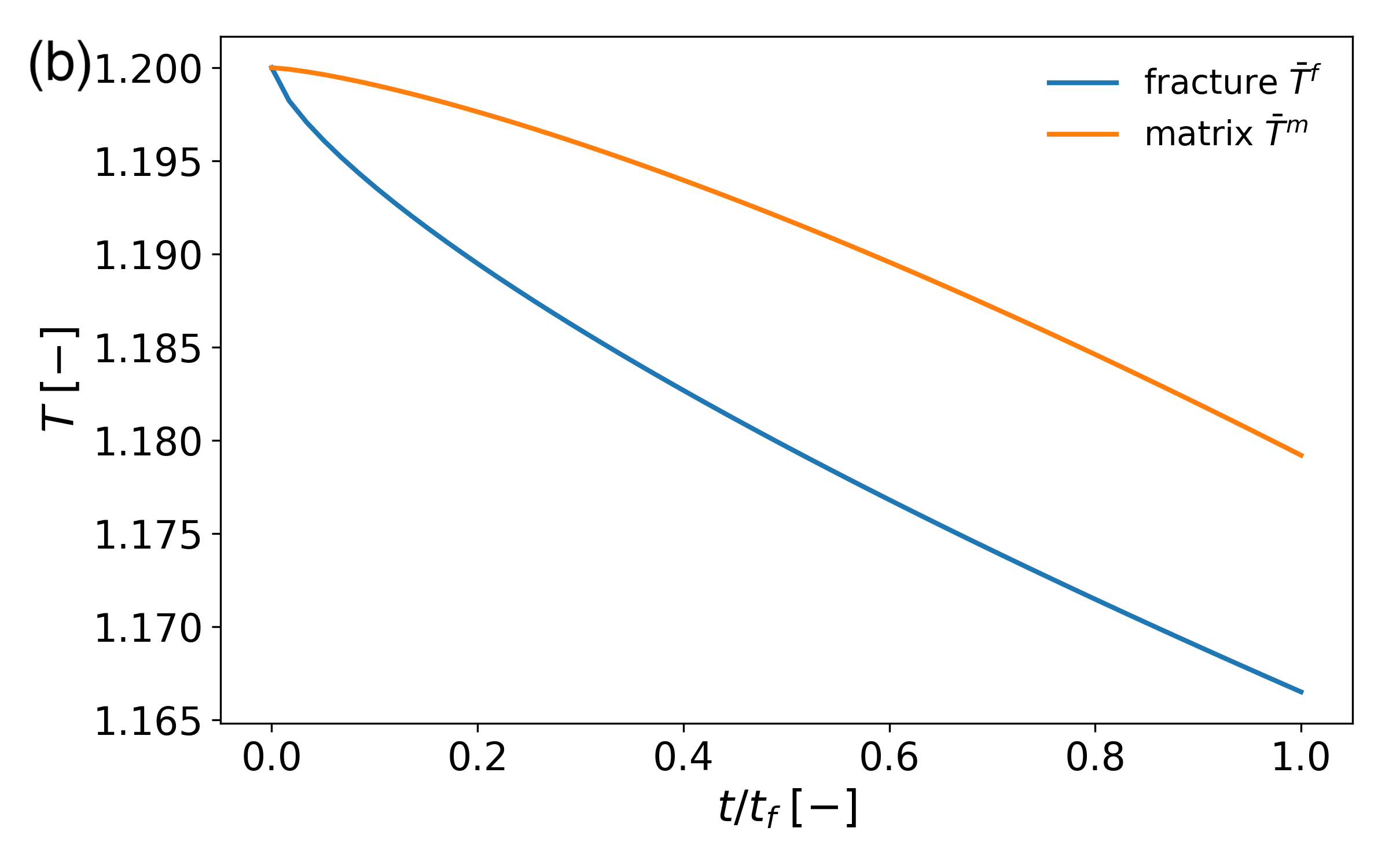}
\includegraphics[width=0.32\textwidth]{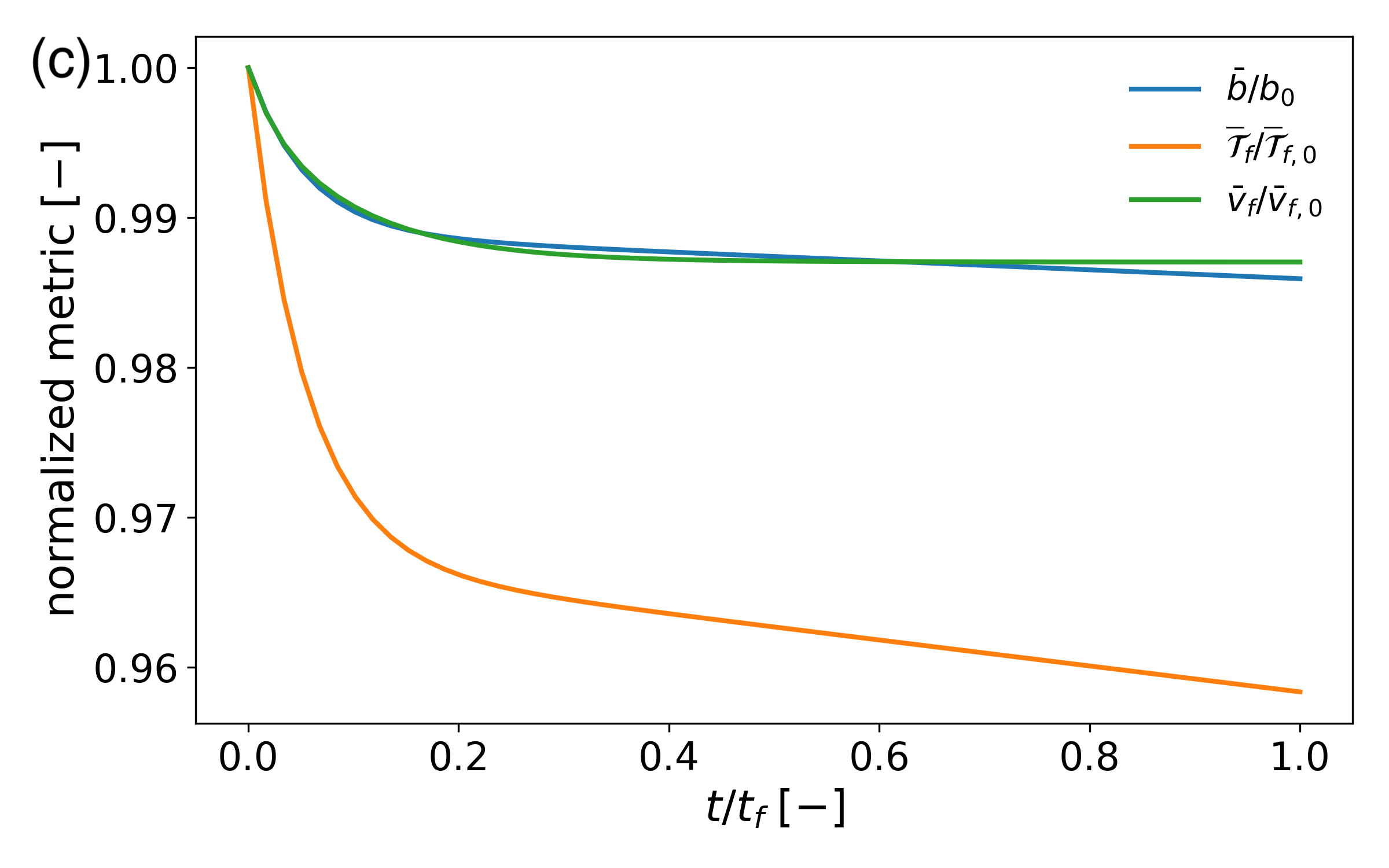}
\vspace{0.3cm}
\includegraphics[width=0.32\textwidth]{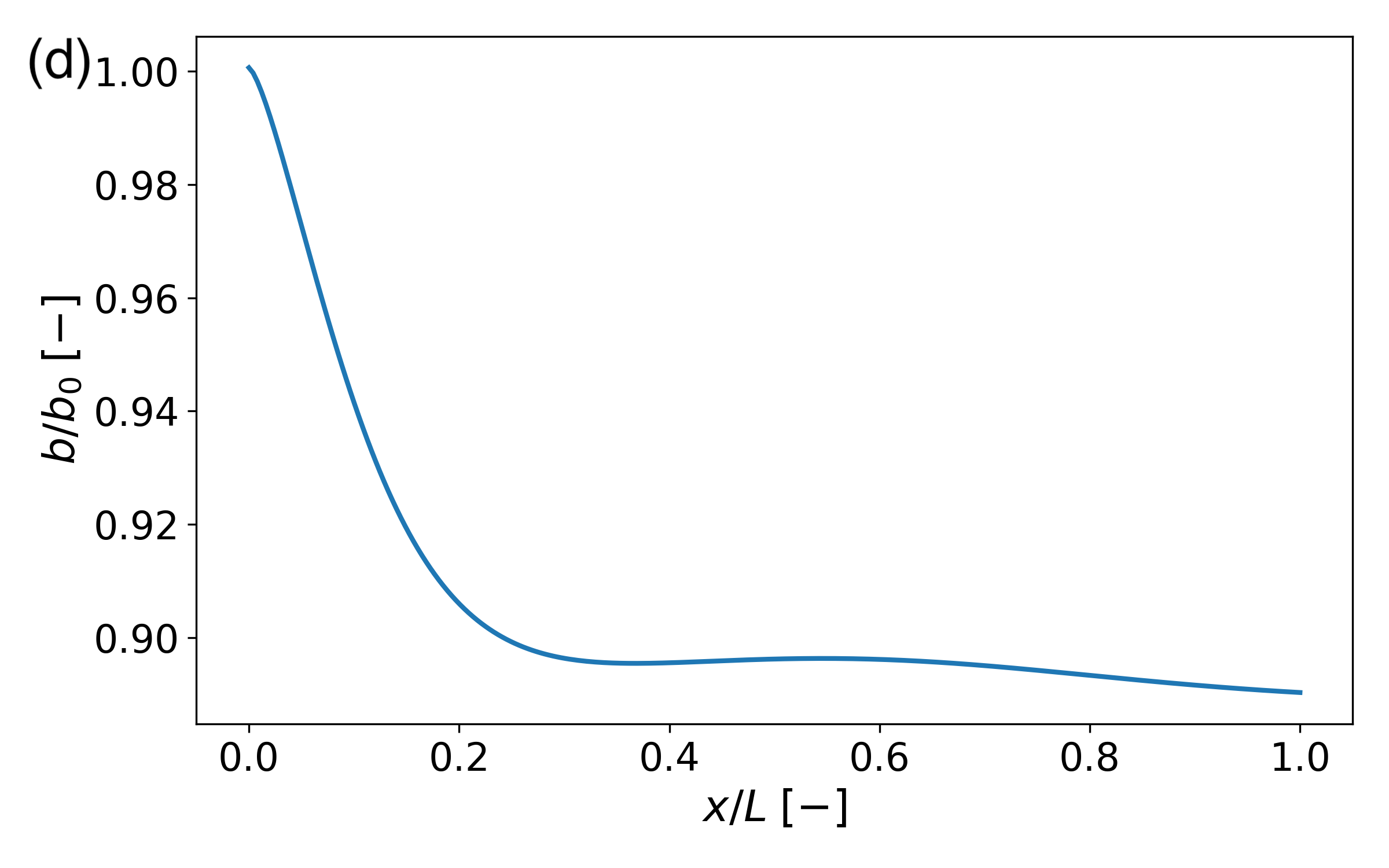}
\includegraphics[width=0.32\textwidth]{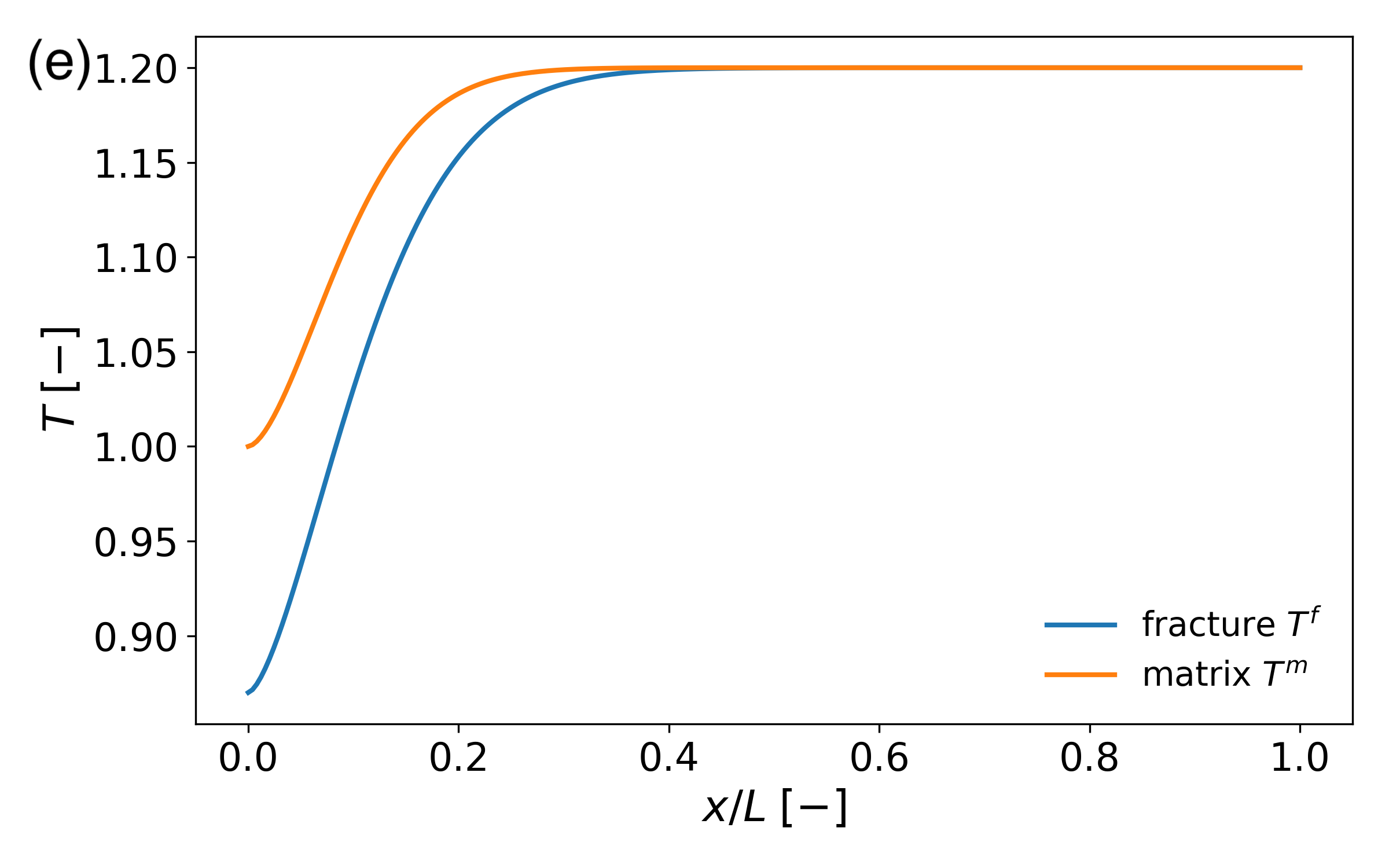}
\caption{
Benchmark~C: fractured thermal-front benchmark with prescribed aperture feedback.
(a) Time histories of the maximum local defects in matrix and fracture, showing an early amplification once aperture feedback is activated.
(b) Mean dimensionless temperatures in matrix and fracture versus normalized time.
(c) Normalized histories of the mean aperture \(\bar b/b_0\), mean fracture transmissivity \(\overline{\mathcal T_f}/\mathcal T_f(b_0)\), and mean fracture velocity magnitude \(\overline{\|\vt^f\|}/\|\vt^f\|_0\).
(d) Final fracture-aperture profile \(b/b_0\) versus normalized position.
(e) Final dimensionless temperature profiles in fracture and matrix versus normalized position.
Aperture feedback modifies the thermo-hydraulic trajectory of the fracture continuum and thereby amplifies and reshapes the defect history. Large defects here simply indicate that the system has moved well outside the exact regime, so the projected closure is being exercised in the nonexact setting it is designed to diagnose.
}
\label{fig:benchmarkC_all}
\end{figure*}
%%%%%%%%%%%%%%%%%%%%%%%%%%%%%%%%%%%%%%%%

\subsection{Benchmark B: fractured thermal front with fixed aperture}
\label{subsec:benchmark_B_num}

The fixed-aperture fractured benchmark is summarized in Fig.~\ref{fig:benchmarkB_all}, which collects the defect history, mean temperatures, final temperature profiles, final saturation profiles, and the reduced geometry and boundary conditions.

The final-profile and time-series diagnostics are summarized in
Table~\ref{tab:benchmarkB_summary}.
All quantities in Table~\ref{tab:benchmarkB_summary} are nondimensional
reduced parameters unless otherwise stated.

\begin{table}[htbp]
\caption{Nondimensional parameters used in Benchmark B.}
\label{tab:benchmarkB_summary}
\centering
\scriptsize
\begin{tabular}{llll}
\hline
Quantity & Symbol & Value & Unit/status \\
\hline
Domain length & \(L\) & \(1\) & - \\
Grid points & \(N_x\) & \(200\) & count \\
Final time & \(t_f\) & \(1\) & - \\
Time step & \(\Delta t\) & \(2\times10^{-3}\) & - \\
Initial saturation & \(\s_0\) & \((0.25,0.25)\) & -\\
Initial temperatures & \(T_0^m,T_0^f\) & \(1.20,1.20\) & - \\
Fracture inlet temperature & \(T_{\rm in}^f\) & \(0.85\) & - \\
Matrix inlet temperature & \(T_{\rm in}^m\) & \(1.15\) & - \\
Fracture inlet saturation & \(\s_{\rm in}^f\) & \((0.45,0.15)\) & -\\
Matrix inlet saturation & \(\s_{\rm in}^m\) & \((0.28,0.22)\) & - \\
Fracture velocity & \(v_f\) & \(0.35\) & - \\
Matrix velocity & \(v_m\) & \(0.02\) & - \\
Aperture & \(b\) & \(10^{-3}\) & - \\
Fracture transmissivity & \(\mathcal T_f\) & \(b^3/12\) & - \\
Heat exchange coefficient & \(\eta_h\) & \(2.0\) & - \\
Saturation exchange coefficient & \(\eta_s\) & \(0.5\) & - \\
Matrix diffusion coefficient & \(D_m\) & \(10^{-3}\) & - \\
Fracture diffusion coefficient & \(D_f\) & \(3\times10^{-3}\) & - \\
Derivative steps & \(\epsilon_s,\epsilon_T\) & \(10^{-6},10^{-4}\) & FD increments \\
\hline
\end{tabular}
\end{table}

The dominant contribution to the loss of exactness is mixed rather than
saturation-sector. Figure~\ref{fig:benchmarkB_all}(a) shows that the
saturation-sector defect remains negligible, while the mixed defect remains
finite in both continua. The time-series diagnostics give \(\max\mathcal D_{sT}^f=0.1845\) and \(\max\mathcal D_{sT}^m=0.1840\), whereas the saturation-sector defect remains zero. Benchmark~B therefore realizes dynamically the same structure isolated in
A2: an isothermal slice-wise test would indicate exactness, but the full
nonisothermal state-domain diagnostic detects mixed nonexactness.

The thermal plots explain that separation.
Figure~\ref{fig:benchmarkB_all}(b) shows that the fracture temperature departs from the initial state more rapidly than the matrix temperature.
The final temperature profiles in Fig.~\ref{fig:benchmarkB_all}(c) show the same contrast spatially: the fracture undergoes the stronger thermal excursion, whereas the matrix profile remains smoother and less perturbed.

The saturation profiles in Figs.~\ref{fig:benchmarkB_all}(d) and \ref{fig:benchmarkB_all}(e) complete the picture.
The fracture samples a wider range of saturations than the matrix, with
\(s_1^f\) reaching \(0.4489\) and \(s_2^f\) decreasing to \(0.1504\). The
matrix response is weaker and smoother.
The two continua therefore sample different trajectories in the augmented state domain even though the aperture is fixed.

Figure~\ref{fig:benchmarkB_all}(f) summarizes the reduced geometry and boundary conditions used in the benchmark and serves only to orient the reader physically.

Benchmark~B thus provides the transition from state-domain verification to fractured dynamics.
Even without aperture feedback, nonisothermal forcing and matrix--fracture exchange already produce continuum-dependent defect histories dominated by the mixed sector.

%%%%%%%%%%%%%%%%%%%%%%%%%%%%%%%%%%%%%%%%
\begin{figure}[t]
\centering
\includegraphics[width=0.95\columnwidth]{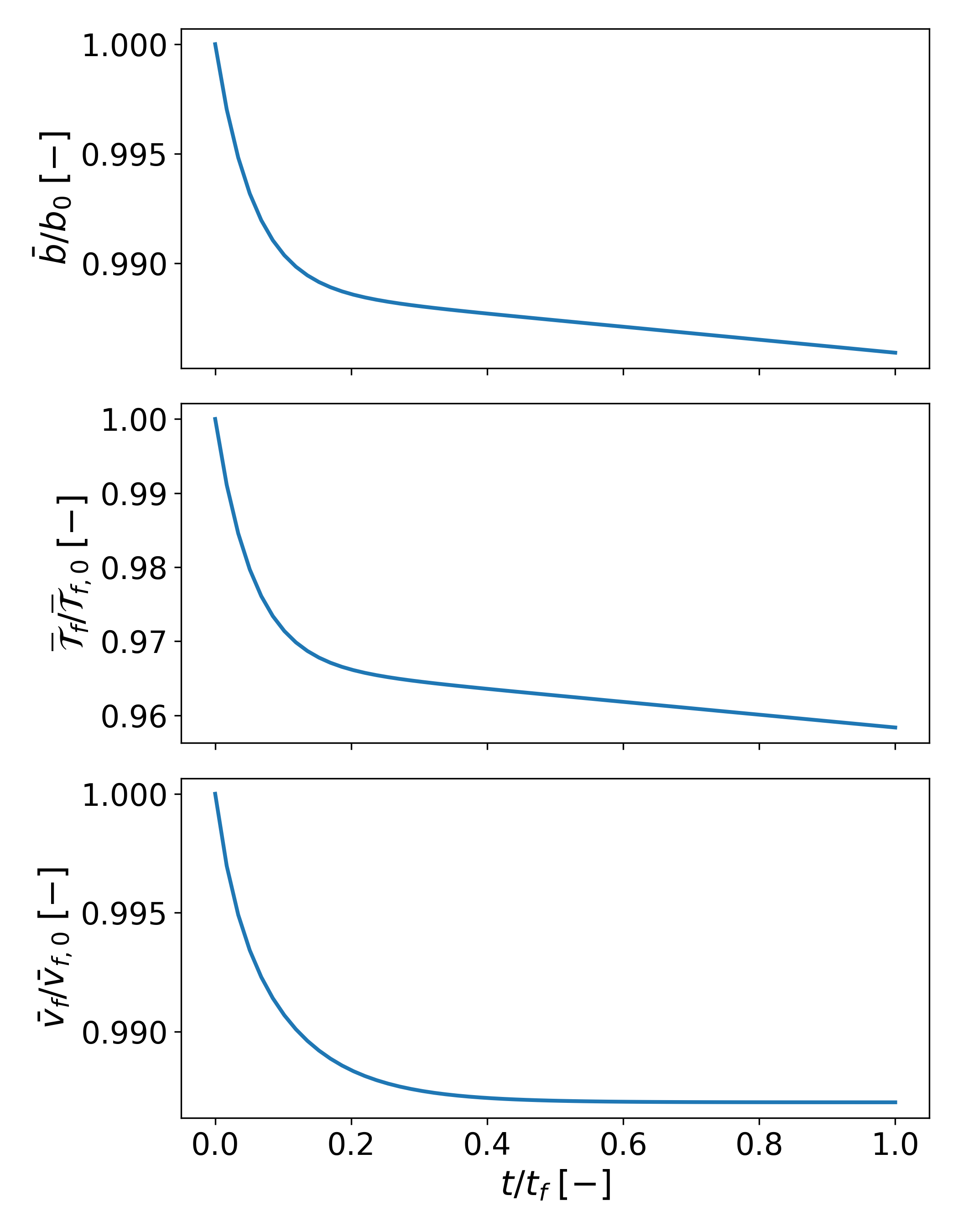}
\caption{
Benchmark~C: aperture-feedback diagnostics on separate axes. From top to bottom: mean normalized aperture \(\bar b/b_0\), mean normalized fracture transmissivity \(\overline{\mathcal T_f}/\mathcal T_f(b_0)\), and mean normalized fracture velocity magnitude \(\overline{\|\vt^f\|}/\|\vt^f\|_0\), all versus normalized time \(t/t_f\). The separate axes make their relative amplitudes easier to compare.
}
\label{fig:benchmarkC_metrics_split}
\end{figure}
%%%%%%%%%%%%%%%%%%%%%%%%%%%%%%%%%%%%%%%%

\subsection{Benchmark C: fractured thermal front with prescribed aperture feedback}
\label{subsec:benchmark_C_num}

The aperture-feedback benchmark is summarized in Fig.~\ref{fig:benchmarkC_all}, which collects the defect history, temperature history, normalized aperture/transmissivity/velocity histories, and the final temperature and aperture profiles.
The same aperture-feedback diagnostics are shown again in Fig.~\ref{fig:benchmarkC_metrics_split} on separate axes to make their relative magnitudes easier to read.

The aperture-feedback setup and profile ranges are summarized in
Table~\ref{tab:benchmarkC_summary}. The principal defect values are reported
in the text below.

\begin{table}[htbp]
\caption{Parameters used in the aperture-feedback diagnostics of Benchmark C.}
\label{tab:benchmarkC_summary}
\centering
\scriptsize
\begin{tabular}{llll}
\hline
Quantity & Symbol & Value & Unit/status \\
\hline
Time samples & \(N_t\) & \(60\) & count \\
Spatial samples & \(N_x\) & \(250\) & count \\
Initial aperture & \(b_0\) & \(6.75\times10^{-4}\) & - \\
Initial mean velocity & \(v_{f,0}\) & \(0.120\) & -\\
Fracture mean temperature range & \(\langle T^f\rangle\) & \(1.1665\)--\(1.2000\) & - \\
Matrix mean temperature range & \(\langle T^m\rangle\) & \(1.1792\)--\(1.2000\) & - \\
Final aperture range & \(b\) & \(6.01\)--\(6.75\)\(\times10^{-4}\) & - \\
Fracture transmissivity & \(\mathcal T_f\) & \(b^3/12\) & - \\
\hline
\end{tabular}
\end{table}

Compared with Benchmark~B, the defect history changes qualitatively.
Figure~\ref{fig:benchmarkC_all}(a) shows a strong early peak in the fracture
defects before relaxation, while the matrix defects are no longer negligible.
The time-series data give a peak fracture mixed defect of \(33.2609\), followed
by a final value of \(1.1020\). Aperture feedback therefore amplifies the
excursion away from exactness, especially in the fracture continuum.

The thermal behavior remains consistent with that picture.
Figure~\ref{fig:benchmarkC_all}(b) shows that the fracture mean temperature departs more rapidly from the initial state than the matrix mean temperature.
The final temperature profiles in Fig.~\ref{fig:benchmarkC_all}(e) show the
same separation spatially: the fracture spans
\(T^f=0.8700\)--\(1.2000\), whereas the matrix spans
\(T^m=1.0000\)--\(1.2000\).

The feedback mechanism becomes explicit in the aperture-related diagnostics.
Figure~\ref{fig:benchmarkC_all}(c) shows the normalized histories of the mean aperture, mean fracture transmissivity, and mean fracture velocity magnitude.
All three quantities change from their initial values, but not at the same rate.
The transmissivity responds more strongly than the aperture itself because of the cubic-law dependence \(\mathcal T_f(b)\propto b^3\), while the fracture velocity records the associated transport response.
The same trend is resolved more clearly in Fig.~\ref{fig:benchmarkC_metrics_split}, where the three histories are plotted on separate axes.

The final aperture profile in Fig.~\ref{fig:benchmarkC_all}(d) remains
spatially heterogeneous, with values between \(6.01\times10^{-4}\) and
\(6.75\times10^{-4}\). Through the cubic transmissivity law, this produces a
fracture transmissivity range from \(1.81\times10^{-11}\) to
\(2.57\times10^{-11}\). Thus, even modest aperture variation is amplified in
the hydraulic response.

A large defect does not invalidate this benchmark.
It indicates that the system has moved well outside the exact regime and that the projected closure is being exercised precisely in the nonexact situation it is meant to diagnose.

Aperture evolution does not change the exactness criterion itself, which remains governed by \eqref{eq:compat_ss}--\eqref{eq:compat_sT}, but it does change the trajectory by which the fractured system moves through the augmented state domain.
That is why the defect history is not only larger than in Benchmark~B, but also qualitatively reshaped by the thermo-hydraulic feedback loop.
The benchmarks reported here validate the augmented-state exactness diagnostics
and illustrate how thermal transport, matrix--fracture exchange, and prescribed
aperture feedback move the system through different regions of state domain. 
They are not
intended as a full production-solver validation of the projected
global-pressure surrogate against an independent phase-pressure simulator. Such
a comparison would require a separate phase-pressure implementation with
identical constitutive data, pressure and phase-flux comparisons, saturation
and temperature comparisons, mass and energy balance errors, and grid/time-step
convergence tests. These steps are outside the scope of the present diagnostic
validation and are left for future work.
Benchmark~C closes the numerical argument: the theorem identifies the
criterion, Benchmark~A verifies it directly, Benchmark~B shows how mixed
nonexactness arises dynamically under fixed aperture, and Benchmark~C shows how
aperture feedback amplifies that loss of exactness in a fractured setting.

\section{Conclusion}
\label{sec:conclusion}

We have formulated the nonisothermal extension of the TD/gTD exactness principle for global-pressure reductions of multiphase Darcy flow.
The key result is that the global-pressure reduction requires path independence
of the mobility-weighted capillary contribution in the full
saturation--temperature state domain. In mathematical terms, this is the
zero-curl compatibility condition for the capillary differential, rather than a
set of separate saturation-space conditions at fixed temperature.  
This yields two compatibility requirements: the classical saturation-sector condition and a genuinely nonisothermal mixed saturation--temperature condition.

The numerical results support this structure directly.
Benchmark~A verifies the three regimes predicted by the theory: fully exact behavior, slice-wise exact but nonisothermally nonexact behavior, and fully nonexact behavior.
In particular, the mixed-defect case confirms that a system may remain integrable on each fixed-temperature slice while still failing to admit a globally exact nonisothermal reduction.

We then embedded the theory in a minimal matrix--fracture model with heat transport, matrix--fracture thermal exchange, and reduced aperture evolution.
The fixed-aperture benchmark shows that temperature evolution alone can make the exactness condition dynamically relevant, with stronger defect histories in the fracture than in the matrix.
The aperture-feedback benchmark further shows that transmissivity feedback modifies the trajectory by which the coupled system probes the augmented state domain, and therefore changes the time history of exactness loss and recovery.

The projected closure should therefore be interpreted with a precise
limitation. It extracts the nearest gradient component of the saturation-sector
capillary field on fixed-temperature slices and preserves the conservative
structure of the reduced balance laws. It does not remove the mixed
saturation--temperature defect and does not restore exact equivalence to the
phase-pressure formulation in nonexact regimes.

The present numerical tests provide verification and diagnostic validation of
the augmented-state exactness criterion and show how reduced matrix--fracture
transport and aperture-sensitive transmissivity drive the system through
different exactness regimes. They do not constitute a full production-solver
validation against an independent phase-pressure implementation, nor do they
address discrete fracture networks or local thermal nonequilibrium models,
both of which are important in geothermal and fractured-reservoir applications
\cite{Medici2023DFNGeothermal,Lei2023EGSDFNTHM,Kostelecky2026LTNE}. Those
extensions require dedicated phase-pressure reference calculations,
grid/time-step convergence studies, and larger fracture-network geometries.

Within these limits, the work provides a compact framework for
nonisothermal global-pressure modeling in fractured multiphase flow: an exact
augmented-state criterion, a reduced fractured-flow setting in which that
criterion becomes operational, and a conservative projected closure for the
nonexact saturation sector.

% ============================================================
\begin{acknowledgments}
Ch.T. would like to acknowledge the support provided by the Deanship of Research (DOR) at King Fahd University of Petroleum \& Minerals (KFUPM) for funding this work through Early Research Career (ERC) grant No.\ EC251017.
\end{acknowledgments}

\bibliographystyle{apsrev4-2}
\bibliography{main}

\end{document}